\documentclass[11pt,journal,draftcls,onecolumn,peerreviewca]{IEEEtran}
\usepackage{amsfonts}
\usepackage{amsmath}
\usepackage{amssymb}
\usepackage{bm}
\usepackage{xcolor,graphicx,float}
\usepackage{color, soul}
\usepackage{stfloats}
\usepackage[numbers,sort&compress]{natbib}
\usepackage[amsmath,thmmarks]{ntheorem}
\usepackage{theorem}
\usepackage{algorithm}
\usepackage{algorithmic}
\usepackage{graphicx}
\usepackage{subfigure}

\usepackage{pict2e}
\usepackage{comment}
\usepackage{makecell}

\newtheorem{lemma}{Lemma}

\newtheorem{proposition}{Proposition}

\newtheorem{proof}{Proof}

\theoremheaderfont{\sc}\theorembodyfont{\upshape}
\theoremstyle{nonumberplain}
\theoremseparator{}
\theoremsymbol{\rule{1ex}{1ex}}

\begin{document}
\title{Joint Multi-Targets Detection and Tracking with mmWave Radar}
\author{Jiang Zhu, Menghuai Xu, Ruohai Guo, Fangyong Wang, Guangying Zheng and Fengzhong Qu
\thanks{This work is supported by the Zhejiang Provincial Natural Science Foundation of China under Grant LY22F010009, National Natural Science Foundation of China under Grant 62371420 and 61901415, National Natural Science Fund for Distinguished Young Scholars under Grant 62225114. (\emph{Corresponding Author: Ruohai Guo.})
	
	Jiang Zhu, Ruohai Guo,  Menghuai Xu and Fengzhong Qu are with the Ocean College, Zhejiang University, No.1 Zheda Road, Zhoushan, 316021, China (email: \{jiangzhu16, menghuaixu, rhguo, jimqufz\}@zju.edu.cn).
	
	Fangyong Wang is with the Hanjiang national Laboratory, Wuhan, 430000, China (email: hzwfy@163.com).

    Guangying Zheng is with the Science and Technology on Sonar Laboratory, Hangzhou, 310023, China  (email: 18069844401@163.com).
	}}
%\date{April 4, 2017}
\maketitle

\begin{abstract}
	Accurate targets detection and tracking with mmWave radar is a key sensing capability that will enable more intelligent systems, create smart, efficient, automated system. This paper proposes a sequential pipeline framework, i.e., detection-association-tracking, named MNOMP-SPA-KF consisting of the target detection and estimation module, the data association (DA) module and the target tracking module. In the target estimation and detection module, a low complexity, super-resolution and constant false alarm rate (CFAR) based two dimensional multisnapshot Newtonized orthogonal matching pursuit (2D-MNOMP) is designed to extract the multi-targets's radial distances and velocities, followed by the conventional (Bartlett) beamformer to extract the multi-targets's azimuths. In the DA module, a sum product algorithm (SPA) is adopted to obtain the association probabilities of the existed targets and measurements by incorporating the radial velocity information. The Kalman filter (KF) is implemented to perform target tracking in the target tracking module by exploiting the asymptotic distribution of the estimators. To improve the detection probability of the weak targets, extrapolation is also coupled into the MNOMP-SPA-KF. Numerical and real data experiments demonstrate the effectiveness of the MNOMP-SPA-KF algorithm, compared to other benchmark algorithms.
\end{abstract}

\section{Introduction}
Target detection and tracking play pivotal roles across various domains, including computer vision, robotics, military and defense, surveillance and security, and autonomous systems. Undoubtedly, the efficacy of target detection and tracking applications is based on the paramount considerations of accuracy and precision. Although sensors like passive infrared (PIR) and time of flight (TOF) bring their unique advantages, they are not without limitations that can impact their performance. For example, PIR sensors, adept at detecting infrared radiation emitted by objects, offer a cost-effective and straightforward implementation. However, their measurements are sensitive to temperature fluctuations, rendering them susceptible to false alarms triggered by environmental changes such as shifts in sunlight and heat sources. However, TOF sensors deliver highly accurate distance information but struggle with challenges such as false alarms due to interference from reflective surfaces and susceptibility to environmental changes such as darkness, smoke, or fog \cite{DG01000}.

Undoubtedly, millimeter-wave (mm-wave) radar technology offers numerous advantages, rendering it highly suitable for diverse applications. MM-wave radar excels at providing precise measurements of the position and velocity of objects, including individuals, within a specified area. In contrast to PIR or TOF sensors, which may be influenced by environmental conditions such as rain, dust, and smoke, mm-wave radar maintains optimal performance even in adverse weather conditions. This reliability extends to both outdoor and indoor environments. Furthermore, mm-wave radar sensors exhibit the capability to operate seamlessly in complete darkness or bright daylight. This versatility makes mm-wave radar an ideal solution for applications that require uninterrupted monitoring, regardless of the time of day \cite{MTT for network}.

Although mm-wave radar has unique advantages over PIR or TOF sensors in target sensing and tracking, it encounters challenges, particularly amidst strong background clutters. In such scenarios, the reflected signal from weaker targets, such as pedestrians with a small radar cross-section (RCS), often gets overshadowed by the reflected waves from formidable targets like surrounding buildings and trucks. Consequently, these weaker targets may go unnoticed. Additionally, the widely recognized issue lies in the estimation and detection algorithm, which tends to generate numerous false alarms, imposing a substantial computational burden to manage and alleviate their impact on overall system performance. 

To tackle the first challenge, a low-complexity, super-resolution compressed sensing (CS) based algorithm incorporating constant false alarm rate (CFAR) should be developed. This addresses target-masking effects, enabling the estimation and detection of weak targets near strong targets. Addressing the second challenge involves the integration of the detection, estimation, and tracking algorithms. This integration ensures that false alarms generated in the detection stage are significantly suppressed during the tracking stage, by allowing a higher false alarm probability in the detection stage. Both serve as the motivation behind our work.

\subsection{Related Works}
The joint target detection, estimation, and tracking algorithm encompasses the extraction of targets' states, comprising azimuth, radial distance, and radial velocity, from the baseband data. This algorithm proceeds by conducting target detection to ascertain the number of targets, executing data association among the targets' states across various time instants, and employing the Kalman filter (KF) for single target tracking. Given the ubiquity of the KF as a standard algorithm, our focus lies predominantly on detailing the relevant works pertaining to target estimation and detection, as well as the data association (DA) algorithm. Additionally, Texas Instruments (TI) has presented a system solution for people tracking and counting utilizing mmWave radar sensors \cite{TIpeopleTracking, TICountandTrack}, and this solution will also be elaborated upon in this discussion.

\subsubsection{Targets Detection and Estimation}
The linear frequency modulated continuous wave (LFMCW) mmWave radar operates by transmitting a chirp signal, and its receiver, equipped with a uniform linear array (ULA), captures the signal reflected by targets within the field of view. Following the dechirping process, the received signal can be conceptualized as a three-dimensional line spectral. Extracting frequencies from this line spectral allows for the estimation of the target's states. Consequently, our focus in this context is on reviewing the pertinent literature regarding line spectral estimation and detection.

The fast Fourier transform (FFT) algorithm \cite{DFT-based Estimation} and  the constant false alarm rate (CFAR) detector \cite{FundamentalsRadarSP} are widely recognized as the most commonly used method for extracting frequencies and target detection, respectively. Despite their simplicity and effectiveness, they do come with limitations. The estimation accuracy of FFT is constrained by the grid spacing, limiting its performance \cite{Modelmismatch}. In addition, the CFAR detector faces challenges where a weak target near a strong target might be missed due to the strong sidelobe of the latter. To overcome these issues, compressive sensing (CS) methods have been introduced. These encompass variational line spectra estimation (VALSE) \cite{VALSE}, atomic norm soft thresholding (AST) \cite{AST}, Newtonized orthogonal matching pursuit (NOMP) \cite{MadhowNOMP}, CFAR-based NOMP (NOMP-CFAR) \cite{NOMP_CFAR}, among others. These methods leverage the sparse nature of targets in the range, Doppler, and spatial domains. Among them, NOMP stands out for its speed but necessitates knowledge of the noise variance. NOMP-CFAR integrates an adaptive CFAR detector into NOMP, eliminating the need for noise variance information but at the cost of increased computational complexity due to forward-backward steps. In contrast, this paper proposes a low-complexity NOMP-CFAR based scheme that is fast and does not require knowledge of the noise variance.
\subsubsection{Data Association}

The challenge of data association and target tracking revolves around extracting target trajectories based on sequences of sensor measurements \cite{FundofObjTrack}. The nearest neighbors (NN) algorithm \cite{NN}, a classical tracking approach, achieves target tracking by matching the current measurement closest to the predicted measurement generated by the target. While it boasts speed and simplicity, it overlooks uncertainties in data association, performing suboptimally in low SNR or cluttered scenarios. The probabilistic data association (PDA) algorithm \cite{PDA}, calculating the average of measurements weighted by association probability, also offers speed and simplicity and excels in scenarios with a single target. However, its performance suffers in scenarios with multiple targets. Addressing this limitation, the joint probabilistic data association (JPDA) algorithm \cite{JPDA, trackerJPDA} computes joint marginal associate probabilities, outperforming NN and PDA in scenarios with multiple targets but making mistakes in complex situations.

The belief propagation (BP) method, also known as the sum-product algorithm (SPA), offers an efficient and scalable solution for data association problems. Operating through ``message passing'' along the edges of the factor graph representing the statistical model, SPA-based multi-target methods have been proposed \cite{WilliamsTAES, SPA1, SPA, SelfTuningwithBP, ScalableDTofGEO, GNNenhanceBP}. SPA provides ``soft'' associations between targets and measurements, featuring advantages of low computational complexity and implementation flexibility. This paper adopts the SPA for DA while incorporating the radial velocity information to calculate the association probabilities of the targets, which benefits the data association.

% This algorithm 
% Therefore, we want to incorporate the algorithm of data association into the LSE\&D algorithm with super-resolution performance and CFAR property and design a tracking algorithm for LFMCW mmWave Radar.
\subsubsection{TI's System Solution}
TI has given a system solution for people tracking and counting using mmWave radar sensor \cite{TIpeopleTracking, TICountandTrack}.
This solution uses FFT and Capon beamforming to obtain the detected points with range, velocity, azimuth, and elevation information.
With the measurement data in polar coordinate and track objects as the inputs of the tracking algorithm, each iteration of the algorithm contains following steps: Point cloud tagging, predicting, associating, allocating, updating and reporting.
The process of associations includes gating and scoring, each tracking unit indicates whether one measurement point is close enough (gating), and if it is, to provide the bidding value (scoring), and the point is assigned to the highest bidder.
During the process of allocating, the new tracking unit is allocated from the points not assigned to any existed track by certain criteria.
During the update step, tracks are updated based on the set of associated points and extended Kalman filter. Compared to TI's solution, this paper uses the superresolution algorithm to acquire the finer state information and detect the weaker targets more reliably. In addition, the correlation of the measurements are taken into account by deriving the Cram\'{e}r-Rao bound (CRB), and the advanced SPA algorithm incorporating the radial velocity information is also integrated to improve the tracking performance.

\subsection{Main Contributions}
In this paper, we present a comprehensive algorithm, MNOMP-SPA-KF, designed for mmWave radar systems, specifically addressing joint multi-targets detection, estimation and tracking. This algorithm integrates MNOMP for target estimation and detection, SPA for data association, and KF for target tracking within a unified framework. The approach involves transforming three-dimensional line spectrum estimation and detection into two-dimensional counterparts with multiple measurement vectors. Subsequently, we optimize the two-dimensional NOMP-CFAR \cite{NOMP_CFAR} by significantly reducing its computation complexity through removing the backward steps. A logic variable is introduced to prevent potential oversight of weak targets masked by stronger ones. Additionally, the proposed algorithm streamlines computation by avoiding multiple FFT computations. The next step involves deriving the CRB for two-dimensional positions and establishing the measurement model in Cartesian coordinates. This model serves to determine the likelihood of the target state. The state-of-the-art SPA algorithm is then employed for DA, with the estimated radial velocity of targets incorporated into SPA for deciding target existence. Extrapolation is introduced into MNOMP-SPA-KF to enhance detection probability. Notably, MNOMP-SPA-KF exhibits high flexibility, allowing a higher false alarm rates in the detection stage while suppressing the false alarms in the tracking stage. This flexibility enables detecting a weaker target, compared to MNOMP alone. Finally, numerical and real experiments are conducted to validate the effectiveness of the MNOMP-SPA-KF algorithm, compared to other joint multi-targets detection, estimation and tracking algorithms.

%\textit{Assumptions}:
%We make the following assumptions \cite{SPA}.
%
%1. the states and measurements of targets obey Gaussian distribution.
%
%2. The single-target states evolve independently.
%% \red{independence of measurement, object motion and initial prior density}
%
%3. Each measurement originates from a target or from clutter (false alarm), and it cannot originate from more than one target simultaneously.
%Conversely, one target can generate at most one measurement at a certain time.
%
%4. the baseband signal are corrupted by additive white Gaussian noise.
%
%5. During the observation time of radar, the state of targets is stable.
% 6. line spectrum assumptions: AWGN, detection probability $\bar{\rm P}_{\rm D}$ and false alarm probability $\bar{\rm P}_{\rm FA}$ only depend on the parameters of algorithm.

% \textit{Notation}:
% Let $\mathbf{a}$, $\mathbf{A}$ denote the vector, matrix, respectively.
% For a D-dimension tensor $\mathcal{Y}$, let ${\mathcal{Y}}_{\mathbf{k}}$ denote the ${\mathbf k}$th element of ${\mathcal{Y}}$, where ${\mathbf{k}} \in \mathbb{N}^{D}$ and $\mathbb{N}$ is the set of all natural numbers, and $\tilde{\mathcal{Y}}$ denotes its spectrum. For a set $\mathcal N$, let $|\mathcal N|$ denote its cardinality.
% $\mathcal{N}(\cdot | \boldsymbol{\mu}; \boldsymbol{\Sigma})$ denotes the Gaussian distribution probability density function (PDF) with mean $\boldsymbol{\mu}$ and covariance $\boldsymbol{\Sigma}$.
% (To-do: $\mathbf{Z}^{t}_{1:P}$, $\mathbf{Z}^{1:t}$, $\mathbf{\Sigma}^{t | t - 1}$, $\mathbf{\Sigma}^{t | t}$, superscript and subscript, $\mathcal{F}$ for FFT)

\section{Problem Setup}\label{ProblemSetup}
This section briefly describes both the  measurement model of the LFMCW mmWave radar and the basic multi-targets tracking method.

%\red{This consists of only two parts. The first part is the LFMCW radar model. Maybe give a picture describing the details. Then, give the basic performances such as range/velocity/angle resolution and the maximum range/maximum velocity related to parameters. The second part is to introduce the basic multitarget-tracking model, including data association and kalman filter. Performing joint estimation and tracking is a hard problem, and we resort to a two stage approach, namely, target estimation and detection, target tracking. In the ensuing section, we present the details.}

%To solve the MOT problem by mmWave radar, we decompose the problem into several steps.
%Starting with the introduction of tensor model of LFMCW radar, we can obtain the measurements of targets in the scenario.
%Associate the measurements at different moments, we can estimate the trajectories of targets and accomplish the tracking of the targets.
%In this section, we will introduce the measurement model first and then describe the transition model.
\begin{figure}
	\centering
	\includegraphics[width = 80mm]{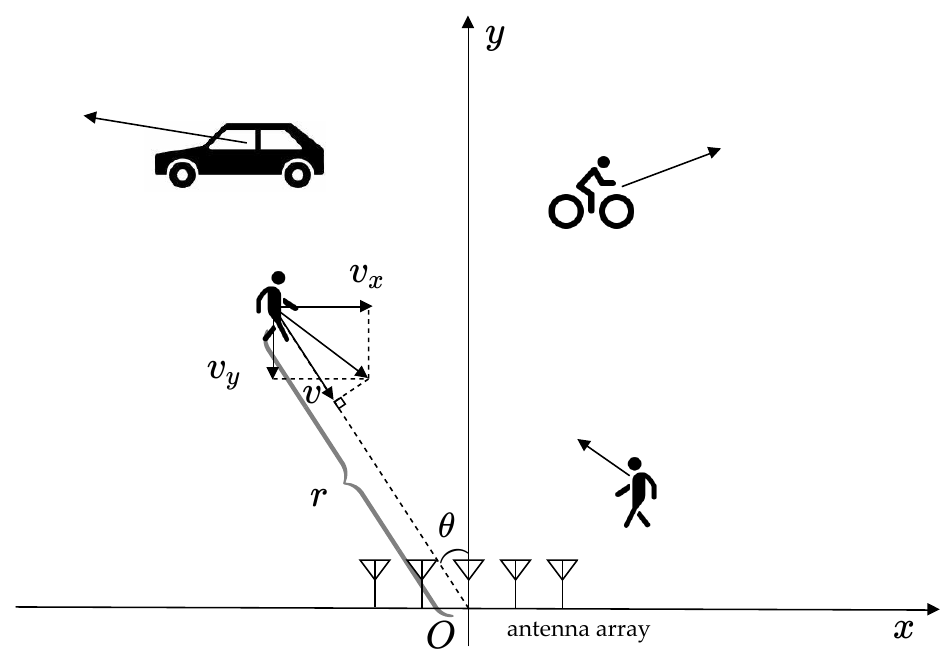}
	\caption{The typical measurement setup of the LFMCW mmWave radar. For simplicity, both the target index and the time index are omitted. }\label{Scenario}
\end{figure}
\subsection{Measurement Model}\label{MeasModel}
Fig.\ref{Scenario} describes the typical measurement setup of the LFMCW radar.  For the LFMCW radar, after dechirping, lowpass filtering and sampling in the $t$th frame, the baseband signal can be described as a three dimensional tensor ${\mathcal Y}(t) \in{\mathbb C}^{N\times M\times L}$, whose $(n, m, l)$th element is
\begin{align}\label{radarmeas}
	&{\mathcal Y}_{n,m,l}(t) =
	\sum_{k = 1}^{K(t)} \gamma_k(t)
	{\rm e}^{{\rm j} (n - 1) \frac{4 \pi}{c} \mu T_{\text s} r_k(t)}
	{\rm e}^{{\rm j} (m - 1) \frac{4 \pi}{c} f_c {T}_{\text{r}} v_k(t)}
	{\rm e}^{{\rm j} (l - 1) \frac{2 \pi}{c} f_c d \sin \theta_k(t)}
	+ \mathcal{\epsilon}_{n, m, l}(t),\notag\\
	&n = 1, \ldots, N, m = 1, \ldots, M, l = 1, \ldots, L,
\end{align}
where $n$, $m$ and $l$ denotes the index of the fast-time domain, slow-time domain and spatial domain, respectively; $N$, $M$ and $L$ are the number of measurements in fast-time domain, slow-time domain and spatial domain, respectively; ${K(t)}$ denotes the number of targets;
$\gamma_k(t)$ is the complex amplitude of the $k$th target;
$r_k(t)$, $v_k(t)$ and $\theta_k(t)$ denote the radial distance, the radial velocity and the azimuth of the $k$th target, respectively; $c$, $\mu$, $T_{\rm s}$, $f_c$, $T_c$, $d$ denote the speed of the electromagnetic wave, the chirp rate, the sampling interval, the carrier frequency, the chirp period, the antenna element spacing, respectively. The antenna element spacing $d$ is set as the half of wavelength $\lambda$, i.e., $d = \lambda / 2$ where $\lambda = c / f_c$; $\mathcal{\epsilon}_{n,m,l}(t)$ is the additive noise and is supposed to be independent and identically distributed (i.i.d.) across the three domains (namely the fast-time domain, slow-time domain and spatial domain) and the time, and $\mathcal{\epsilon}_{n,m,l}(t) \sim {\mathcal {CN}}(0,\sigma_{\rm b}^2)$ with $\sigma_{\rm b}^2$ being the variance of the noise.

For the LFMCW radar, the key concepts are the resolutions and limits of the range, velocity and angle \cite{TImmwave}. Here we briefly introduce those as follows.
\begin{itemize}
	\item The range resolution $r_{\rm res} = \frac{c}{2 B}$ depends only on the bandwidth $B = \mu T_{\rm ramp}$ swept by the chirp, where $T_{\rm ramp}$ is the ramp time of LFMCW radar.
	\item The maximum range $r_{\rm max}$ is limited by the ADC sampling rate $F_{\rm s}=1/T_{\rm s}$ and the chirp rate given by $r_{\rm max} = \frac{c}{2 \mu T_{\rm s}}$. 
	\item The maximum measurable radial velocity $v_{\rm max}=\frac{\lambda}{4T_c}$ is inversely proportional to the chirp interval $T_c$. Thus higher $v_{\rm max}$ requires closely spaced chirps.
	\item The velocity resolution $v_{\rm res} = \frac{\lambda}{2NT_c}$ is inversely proportional to the frame time, i.e., the product of the number of the slow time dimensions and the chirp interval $T_c$.
	\item The maximum field of view $\theta_{\rm max}$ is $\theta_{\rm max}=\sin^{-1}\left(\frac{\lambda}{2d}\right)$. A spacing $d$ of half wavelength $\lambda/2$ results in the largest field of view $\pm 90^{\circ}$.
	\item The angle resolution $\theta_{\rm res}$ is given by $\theta_{\rm res} = \frac{\lambda}{Ld \cos \theta}$ and simplifies to $\theta_{\rm res} = \frac{2}{L\cos \theta}$ with $d=\lambda/2$. This resolution depends on the azimuth angle, reaching its maximum of $\frac{2}{L}$ at $\theta=0^{\circ}$. In engineering, resolution is often calculated to be $\frac{2}{L}$.
\end{itemize}
It is worth noting that those concepts are analyzed through linear signal processing. They provide basic benchmarks for the LFMCW radar, and those criteria could be improved via nonlinear signal processing methods.

%The reference \cite{TImmwave} gives some basic parameters of the LFMCW radar under the traditional model and processing method.
%The observation range of the radial distance is $(0, r_{\rm max})$, where the correspongding maximum radial distance $r_{\rm max} = \frac{c}{2 \mu T_{\rm s}}$.
%
%Similarly, the range of raidal velocity is $\left(- v_{\rm max}, v_{\rm max}\right)$, where $v_{\rm max} = \frac{c}{4 f_{\rm c} T_{\rm r}}$, and the corresponding velocity resolution is $v_{\rm res} = \frac{c}{2 f_c M T_{\rm r}}$.
%The range of azimuth angle can reach $(- \frac{\pi}{2}, \frac{\pi}{2})$ theoretically, which may be limited by the performance of radar.
%The resolution of azimuth angle is $\theta_{\rm res} = \frac{2}{L \cos \theta}$, which depends on the real azimuth angle.
\subsection{Multi-targets Tracking Model}
The constant velocity (CV) motion model stands out as the most widely employed motion model due to its simplicity and its ability to adequately describe the movements of entities such as pedestrians.

Consider a multi-targets scenario in a two dimensional setting. Define the state vector $\mathbf{x}_{k}(t)$ of the $k$th target at the $t$th frame as
\begin{align}
	\mathbf{x}_{k}(t) = [p_{x, k}(t), v_{x, k}(t), p_{y, k}(t), v_{y, k}(t)]^{\rm T} \in {\mathbb R}^4,
\end{align}
where $p_{x, k}(t)$ and $v_{x, k}(t)$ are the position and velocity along the $x$-axis, while $p_{y, k}(t)$ and $v_{y, k}(t)$ are the position and velocity along the $y$-axis. The target state transition model is described as \cite{Barshabook}
\begin{align}\label{TransitionEquation}
\mathbf{x}_{k}(t) = \mathbf{A} \mathbf{x}_{k}(t - 1) + \mathbf{w}_{k}(t),
\end{align}
where
\begin{align}
\mathbf{A} = \begin{bmatrix}
	1 & T & 0 & 0 \\
    0 & 1 & 0 & 0 \\
    0 & 0 & 1 & T \\
    0 & 0 & 0 & 1 \\
\end{bmatrix}
\end{align}
is the state transition matrix, $T$ is the frame interval, $\mathbf{w}_{k}(t)$ denotes the Gaussian noise with zero mean and known covariance matrix ${\mathbf Q}_k(t)$.

Define the $h$th measurement $\mathbf{z}_{h}(t)$ at the  instant $t$ as
\begin{align}\label{meas}
\mathbf{z}_{h}(t) = \left[ p_{x, h}(t), p_{y, h}(t) \right]^{\rm T} \in \mathbb{R}^2,
\end{align}
where $p_{x, h}(t)$ and $p_{y, h}(t)$ are the observed position on the x-axis and y-axis, respectively; $h = 1, \ldots, H(t)$ is the index of measurement, where $H(t)$ is the number of measurements at instant $t$. For a given target (omitting the index of the target), its measurements $\mathbf{z}(t)$ is related to its state $\mathbf{x}(t)$ according to
\begin{align}\label{measmodel}
\mathbf{z}(t) = \mathbf{H} \mathbf{x}(t) + \mathbf{v}(t),
\end{align}
where
\begin{align}
{\mathbf H}=
\begin{bmatrix}
1 & 0 & 0 & 0 \\
0 & 0 & 1 & 0 \\
\end{bmatrix}
\end{align}
is the measurement matrix and $\mathbf{v}(t)$ is the measurement noise following the normal distribution with zero mean and covariance matrix $\mathbf{R}(t)$.

% although we have outlined the standard Kalman filter steps for multiple target tracking,
It is important to note that several questions naturally arise. Firstly, the measurements in eq. (\ref{meas}) correspond to the positions of multi-targets along the $x$ and $y$ axes, which need to be extracted in real-time from the baseband measurement model (\ref{radarmeas}). Additionally, the covariance matrix $\mathbf{R}(t)$ in the measurement model (\ref{measmodel}) is unknown and requires estimation. Secondly, efficient data association is crucial for continuous target tracking. This involves associating measurements with existing targets and establishing new tracks as needed. Thirdly, radial velocities of multiple targets are simultaneously estimated from the baseband data. These radial velocities should also be utilized for data association since the radial velocity is closely related to the state of the target. In the subsequent sections, we will address and resolve these challenges step by step.

\section{MNOMP-SPA-KF}
This section develops the MNOMP-SPA-KF for joint multi-targets estimation, detection and tracking. As shown in Fig. \ref{FlowChart}, the MNOMP-SPA-KF consists of three modules, namely the MNOMP module, the SPA module and the KF module. The MNOMP module processes the baseband data ${\mathcal Y}(t)$ to extract the radial distances, radial velocities and the azimuths of the multi-targets,  the  radial distances and the azimuths of the multi-targets are further processed to obtain the measurements $\{{\mathbf z}_h(t)\}_{h=1}^{H(t)}$ and the measurement covariance matrices $\{{\mathbf R}_h(t)\}_{h=1}^{H(t)}$, which are then input to the SPA and the KF module. The SPA performs data association by using the measurements $\{{\mathbf z}_h(t),{\mathbf R}_h(t)\}_{h=1}^{H(t)}$,  the predicted means and covariance matrices of the targets' state $\{\hat{\mathbf x}_k(t|t-1),{\boldsymbol\Sigma}_k(t|t-1)\}_{k=1}^{\hat{K}(t-1)}$, with the radial velocities $\{\hat{v}_h(t)\}_{h=1}^{H(t)}$ as an additional input, and outputs the data association probability $\beta_{k\rightarrow h}(t)$ which is the probability that the $h$th measurement is associated to the $k$th target. With the data association probability $\beta_{k\rightarrow h}(t)$, the measurements $\{{\mathbf z}_h(t)\}_{h=1}^{H(t)}$ and the measurement covariance matrices $\{{\mathbf R}_h(t)\}_{h=1}^{H(t)}$, the posterior means and covariance matrices of the multi-targets at the $(t-1)$th frame as inputs, the KF module generates the final outputs corresponding to the means and covariance matrices of the targets' state. Below we present these details.

\begin{figure}
	\centering
	\includegraphics[width = 140mm]{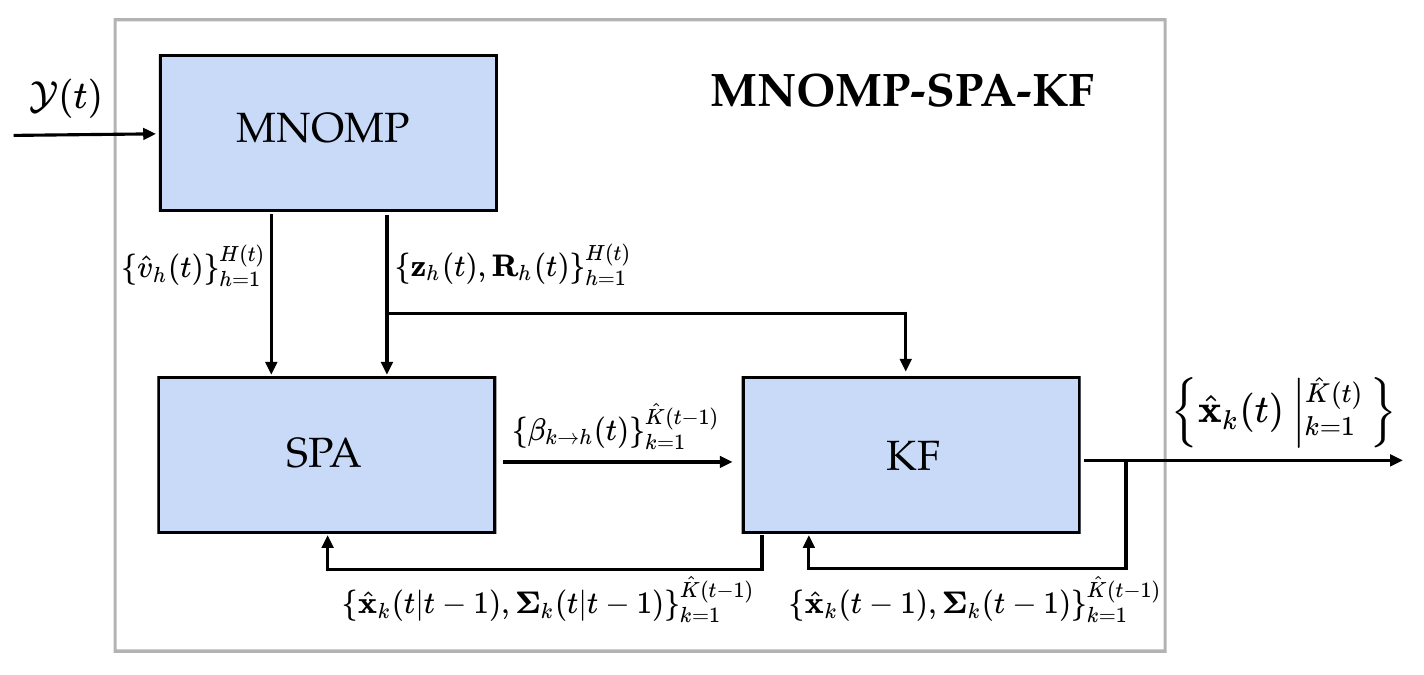}
	\caption{The flow chart of MNOMP-SPA-KF algorithm.}\label{FlowChart}
\end{figure}

\subsection{Superfast 2DMNOMP for Target Sensing}\label{2DMNOMP}

In this subsection, we propose a superfast 2DMNOMP algorithm to extract the radial distances, radial velocities and azimuth angles of multiple targets from the baseband signal (\ref{radarmeas}). Note that the signal model (\ref{radarmeas}) can be formulated as a three dimensional line spectrum estimation (LSE) problem, and the three dimensional CFAR based NOMP (3D-NOMP-CFAR) can be designed to extract the state of multi-targets \cite{NOMP_CFAR}. However, the computation complexity of 3D-NOMP-CFAR is very high due to the multiple implementation of three dimensional FFT \footnote{It is found that 3D-NOMP-CFAR will generate too many false alarms when applied to the real data acquired by the mmWave radar.}. Here we design a low complexity superresolution algorithm with acceptable performance loss. We conduct dimension reduction and formulate the model (\ref{radarmeas}) as  (omitting the time notation $t$)
\begin{align}\label{2dmmv}
&\mathbf{Y}_{l}\left[n, m\right] = \sum_{k = 1}^{K} \gamma_{k, l} {\rm e}^{{\rm j} (n - 1) \omega_{x, k}} {\rm e}^{{\rm j} (m - 1) \omega_{y, k}} + \mathcal{\epsilon}_{n, m, l},
\end{align}
where $n = 1, \ldots, N$, $m = 1, \ldots, M$, $l = 1, \ldots, L$,
\begin{align}\label{gammakldef}
	\gamma_{k, l} = \gamma_k{\rm e}^{{\rm j} (l - 1) \omega_{z, k}},
\end{align}
\begin{subequations}\label{FreqToState}
	\begin{align}
		{\omega}_{x, k} &= \frac{4 \pi}{c} \mu r_k T_{\text s} = 2 \pi \frac{r_k}{r_{\rm max}}, \label{fast time omega} \\
		{\omega}_{y, k} &= \frac{4 \pi}{c} f_c v_k {T}_{\text{r}} = \pi \frac{v_k}{v_{\rm max}}, \label{slow time omega} \\
		{\omega}_{z, k} &= \frac{2 \pi}{c} f_c d \sin \theta_k = \pi \sin \theta_k. \label{spatial omega}
	\end{align}
\end{subequations}
Define the array manifold vector
\begin{align}
	{\mathbf a}_N(\omega)=[1,{\rm e}^{{\rm j}\omega},\cdots,{\rm e}^{{\rm j}(N-1)\omega}]^{\rm T},
\end{align}
model (\ref{2dmmv}) can be formulated as a matrix-vector model
\begin{align}\label{2dmmvmvform}
	&\mathbf{Y}_{l} = \sum_{k = 1}^{K} \gamma_{k, l} 	{\mathbf a}_N(\omega_{x,k}){\mathbf a}_M^{\rm T}(\omega_{y,k}) + {\boldsymbol\epsilon}_{l}.
\end{align}

Consequently, model (\ref{2dmmv}) is a two dimensional LSE model with snapshots $L$. Note that there exists a one to one correspondence between the target state $(r_k,v_k,\theta_k)$ and the frequencies $({\omega}_{x, k},{\omega}_{y, k},{\omega}_{z, k})$. For model (\ref{2dmmv}) and dropping the constraint (\ref{gammakldef}), we extract $({\omega}_{x, k},{\omega}_{y, k})$ to obtain $(r_k,v_k)$. Then we obtain $\gamma_{k, l}$ to estimate the azimuth angle $\theta_k$ by taking the constraint (\ref{gammakldef}) into consideration.

For the data matrix set $\{\mathbf{Y}_l\}_{l=1}^L$ with $L$ snapshots, we can employ 2D-MNOMP-CFAR algorithm to estimate angular frequency $\omega_{x, k}$, $\omega_{y, k}$ and gain vector ${\mathbf g}_{k}=[\gamma_{k, 1},\gamma_{k, 2},\cdots,\gamma_{k, L}]^{\rm T}\in {\mathbb C}^L$. For further details about 2D-MNOMP-CFAR, please refer to \cite{NOMP_CFAR}. To further reduce the computation complexity of 2D-MNOMP-CFAR, we have streamlined the algorithm as follows:
\begin{itemize}
    \item The conventional 2D-MNOMP-CFAR consists of two steps: Initialization and Detection. The initialization stage is to provide enough number of frequency candidates and can be viewed as the forward step. The detection step is to remove the most unlikely frequency based on the CFAR criterion and can be regarded as the backward step. Once the most unlikely frequency is removed, the amplitudes and frequencies of the remaining targets are jointly refined. The backward step is time consuming. A natural way is to directly discarding the backward step and the forward step stops based on the CFAR criterion. However, this may introduce a potential issue known as ``target masking'', i.e., one or more targets are situated among the training cells around the cell under test (CUT), causing an increase in the estimated noise variance and potentially leading to missed detections \cite{FundamentalsRadarSP}. To overcome this issue, we introduce a hyperparameter $K_{\rm invalid}$.
    When detecting the target, its detection statistic may not surpass its threshold due to the ``target masking'' effect. Therefore, the target that does not surpass its threshold will also be added into the candidate target set.  The algorithm stops when the number of the consecutive targets whose detection statistic does not exceed its threshold is $K_{\rm invalid}$. When the algorithm stops, the last $K_{\rm invalid}$ targets will be removed from the candidate target set, yielding the final target detection set. For example, suppose $K_{\rm invalid}=2$. In the current and the next iterations, if the consecutive two targets do not surpass their thresholds, the algorithm stops and the final detected targets do not include the two targets. If in the next iteration, the second target's detection statistic surpasses its threshold, the second target is declared to be detected, then the first target also becomes a valid target.

    \item We employ 2D-FFT only once to obtain ${\mathcal F}_{\omega}({\mathbf Y}_l)$  at the beginning of algorithm. Let ${\mathcal F}_{\omega_x,\omega_y}(\cdot)$ denote the two dimensional FFT evaluated at frequency $(\omega_x,\omega_y)$. Let $\mathcal{K} = \left\{ \left( \hat{\mathbf g}_k, \hat{\omega}_{x, k}, \hat{\omega}_{y, k} \right), k = 1, \ldots, \hat{K} \right\}$ denote the set of estimated parameters, where $\hat{\omega}_{x, k}, \hat{\omega}_{y, k}$ is the angular frequencies estimates, $\hat{\mathbf g}_k = \left[ \hat{\gamma}_{k, 1}, \ldots, \hat{\gamma}_{k, L} \right]^{\rm T}$ is the estimation of gain vector, $\hat{K}$ is the number of detected targets. The residue ${\mathbf{Y}}_{{\rm r}, l}^{\mathcal{K}}$ is
    \begin{align}\label{timeResidual}
        & {\mathbf{Y}}_{{\rm r}, l}^{\mathcal{K}} = {\mathbf{Y}}_l - \sum_{k = 1}^{\hat{K}} \hat{\gamma}_{k, l} {\mathbf{a}}_{N} \left( \hat{\omega}_{x, k} \right) {\mathbf{a}}_{M}^{\rm T} \left( \hat{\omega}_{y, k} \right). \nonumber \\
        & l = 1, \ldots, L.
    \end{align}
    To perform new signal detection, the spectrum of the residue ${\mathcal F}_{\omega_x,\omega_y}({\mathbf{Y}}_{{\rm r}, l}^{\mathcal{K}})$ is calculated as 
    \begin{align}\label{timeResidual}
        {\mathcal F}_{\omega_x,\omega_y}({\mathbf{Y}}_{{\rm r}, l}^{\mathcal{K}})
		&= {\mathcal F}_{\omega_x,\omega_y}({\mathbf Y}_l) - \sum_{k = 1}^{\hat{K}} \hat{\gamma}_{k, l}
		{\mathcal F}_{\omega_x,\omega_y}({\mathbf{a}}_{N} \left( \hat{\omega}_{x, k} \right) {\mathbf{a}}_{M}^{\rm T} \left( \hat{\omega}_{y, k} \right)), \nonumber\\
        &= {\mathcal F}_{\omega_x,\omega_y}({\mathbf Y}_l)
		- \sum_{k = 1}^{\hat{K}} \hat{\gamma}_{k, l}
		\frac{1 - {\rm e}^{- {\rm j} N \left( \omega_x - \hat{\omega}_{x, k} \right)}}{1 - {\rm e}^{- {\rm j} \left( \omega_x - \hat{\omega}_{x,k} \right)}}
		\frac{1 - {\rm e}^{- {\rm j} M \left( \omega_y - \hat{\omega}_{y,k} \right)}}{1 - {\rm e}^{- {\rm j} \left( \omega_y - \hat{\omega}_{y,k} \right)}}, \nonumber \\
        & l = 1, \ldots, L,
    \end{align}
demonstrating that ${\mathcal F}_{\omega_x,\omega_y}({\mathbf{Y}}_{{\rm r}, l}^{\mathcal{K}})$ can be calculated efficiently.
\end{itemize}

For the 2D-MNOMP-CFAR algorithm, it outputs the estimates $\{\hat{\mathbf g}_k,\hat{\boldsymbol\omega}_x,\hat{\boldsymbol\omega}_y,\hat{K},\hat{\sigma}_b^2\}$ with $\hat{\sigma}_b^2$ denoting the noise variance estimates. We then estimate the azimuth via least squares method by taking the constraint (\ref{gammakldef}) into account, i.e.,
\begin{align}
	\hat{\omega}_{z, k}, \hat{\gamma}_k = \mathop{\arg \min} \limits_{\omega_{z, k}, \gamma_k} \left\| \gamma_k \mathbf{a}_{L} \left( \omega_{z, k} \right) - \hat{\mathbf g}_k \right\|_2^2.
\end{align}
As a result, the final estimates $\{\hat{K},\hat{\sigma}_b^2,(\hat{\gamma}_k,\hat{\omega}_{x,k},\hat{\omega}_{y,k},\hat{\omega}_{z,k}),k=1,2,\cdots,\hat{K}\}$ are obtained, and the targets' state can be extracted through frequencies via eq. (\ref{FreqToState}).

Note that we have used model (\ref{measmodel}) to perform target tracking. While the estimates we have obtained from the baseband data are radial distances, radial velocities and azimuths, we need to use these to obtain a pseudo measurement model similar to (\ref{measmodel}). As shown in Fig. \ref{Scenario}, the target's radial distance $r$, radial velocity $v$ and azimuth angle $\theta$ are related to the target's state via
\begin{subequations}
\begin{align}
r &= \sqrt{p_x^2 + p_y^2}, \label{pxpyr}\\
\theta &= \arctan \left( p_x / p_y \right), \label{pxpytheta}\\
\label{RadialVelocity} v &= v_x \sin \theta + v_y \cos \theta.
\end{align}
\end{subequations}
Under regularity conditions, it is expected that the distribution of the maximum likelihood estimator asymptotically approaches to a Gaussian distribution with mean equal to the true value and the covariance matrix equal to the CRB. Therefore, given the measurement model (\ref{2dmmv}) and the constraint (\ref{pxpyr}) and (\ref{pxpytheta}), we obtain the CRB ${\rm CRB}([p_x,p_y]^{\rm T})$. For simplicity, we assume a single target scenario and the results are summarized in the following Proposition \ref{PropositionCRBpxpy}.
\begin{proposition}\label{PropositionCRBpxpy}
	For the single target scenario, the ${\rm CRB}([p_x,p_y]^{\rm T})$ is
	\begin{small}
		\begin{align}\label{CRBpxpy}
			{\rm CRB}([p_x,p_y]^{\rm T}) = \frac{6 \sigma^2}{\pi^2 NML |\gamma|^2} \begin{bmatrix}
				\frac{r_{\rm max}^2 \sin^2\theta}{4 (N^2 - 1)} + \frac{r^2 }{L^2 - 1} & \left[\frac{r_{\rm max}^2}{4 (N^2 - 1)} - \frac{r^2}{(L^2 - 1) \cos^2 \theta}\right] \sin\theta \cos\theta \\
				\left[\frac{r_{\rm max}^2}{4 (N^2 - 1)} - \frac{r^2}{(L^2 - 1) \cos^2 \theta}\right] \sin\theta \cos\theta & \frac{r_{\rm max}^2 \cos^2\theta}{4 (N^2 - 1)} + \frac{r^2 \tan^2 \theta}{L^2 - 1} \\
			\end{bmatrix}.
		\end{align}
	\end{small}
\end{proposition}
\begin{proof}
	The proof is postponed into Appendix \ref{CRBpxy}.
\end{proof}
According to \cite{Kayest}, it is expected that the maximum likelihood estimator $[\hat{p}_x^{\rm ML},\hat{p}_y^{\rm ML}]^{\rm T}$ is asymptotically unbiased and approaches to the CRB, i.e.,
\begin{align}
    [\hat{p}_x^{\rm ML},\hat{p}_y^{\rm ML}]^{\rm T}\stackrel{a}{\sim} {\mathcal N}([p_x,p_y]^{\rm T},{\rm CRB}([p_x,p_y]^{\rm T})),
    \label{nomalpxy}
\end{align}
where $\stackrel{a}{\sim}$ stands for asymptotically distributed and ${\rm CRB}([p_x,p_y]^{\rm T})$ is defined in Proposition \ref{PropositionCRBpxpy}. Note that ${\rm CRB}([p_x,p_y]^{\rm T})$ depends on the true positions of the targets which are unknown. In general, when the maximum likelihood estimator approaches to the CRB, $[\hat{p}_x^{\rm ML},\hat{p}_y^{\rm ML}]^{\rm T}$ is very close to $[p_x,p_y]^{\rm T}$, and ${\rm CRB}([\hat{p}_x^{\rm ML},\hat{p}_y^{\rm ML}]^{\rm T})$ approximates ${\rm CRB}([p_x,p_y]^{\rm T})$ well. Consequently, a pseudo measurement 
\begin{align}\label{pseudomeasurement}
    \mathbf{z}_{{\rm ML}}(t)=
    \left[\begin{array}{c}
        {p}_x(t) \\
        {p}_y(t) 
    \end{array}\right]+{\mathbf v}_{{\rm ML}}(t) 
\end{align}
is obtained, where the measurements $\mathbf{z}(t)$ is defined as 
\begin{align}
	\mathbf{z}_{{\rm ML}}(t)=[\hat{p}_x^{\rm ML}(t),\hat{p}_y^{\rm ML}(t)]^{\rm T},
    \label{z_t}
\end{align}
${\mathbf v}_{{\rm ML}}(t)$ is Gaussian distributed satisfying ${\mathbf v}_{{\rm ML}}(t)\sim {\mathcal {CN}}({\mathbf 0},{\rm CRB}([\hat{p}_x^{\rm ML},\hat{p}_y^{\rm ML}]^{\rm T}))$. Since we use the low complexity suboptimal 2D-MNOMP algorithm  whose performance is rather well, we replace the ML estimator with the estimates $[\hat{p}_x(t),\hat{p}_y(t)]^{\rm T}$ obtained via the 2D-MNOMP to yield the following pseudo measurement model 
\begin{align}\label{pseudomeasurementmnomp}
    \mathbf{z}(t)=
    \left[\begin{array}{c}
        {p}_x(t) \\
        {p}_y(t) 
    \end{array}\right]+{\mathbf v}(t),
\end{align}
where 
\begin{align}
	\mathbf{z}(t)=[\hat{p}_x(t),\hat{p}_y(t)]^{\rm T},
    \label{z_t}
\end{align}
 $\mathbf v (t) \sim \mathcal N (\mathbf{0},{\mathbf R}(t))$, and ${\mathbf R}(t)$ is a scaled ${\rm CRB}([p_x,p_y]^{\rm T})$ evaluated at $[\hat{p}_x(t),\hat{p}_y(t)]^{\rm T}$, i.e., ${\mathbf R}(t)=\kappa {\rm CRB}([\hat{p}_x(t),\hat{p}_y(t)]^{\rm T})$ and $\kappa>1$.
\begin{figure}
    \centering
    \subfigure[]{\includegraphics[width = 80mm]{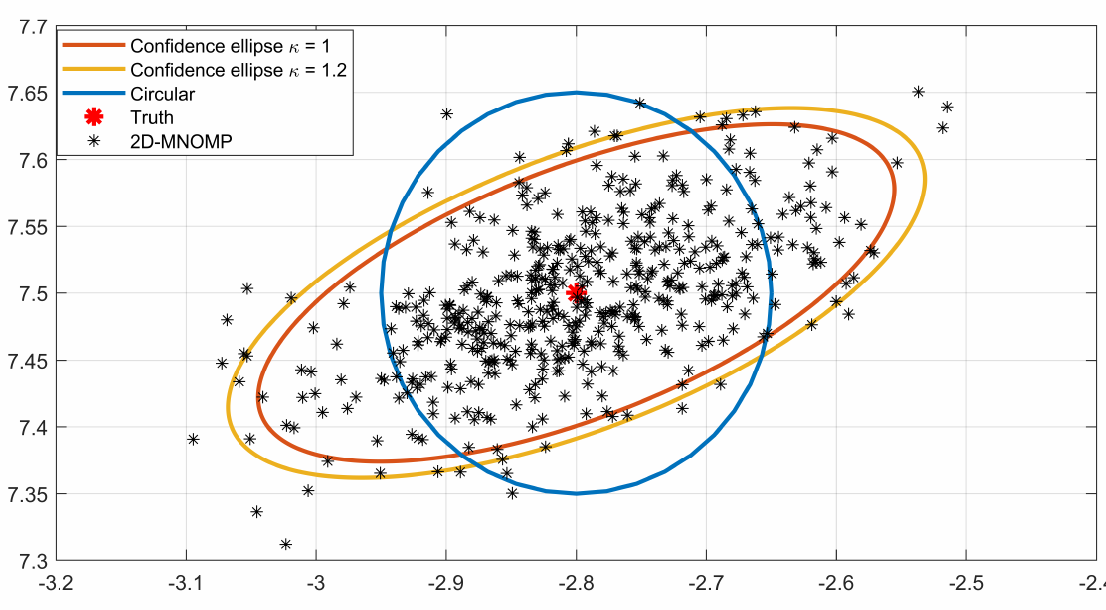}\label{ell1}}
    \subfigure[]{\includegraphics[width = 80mm]{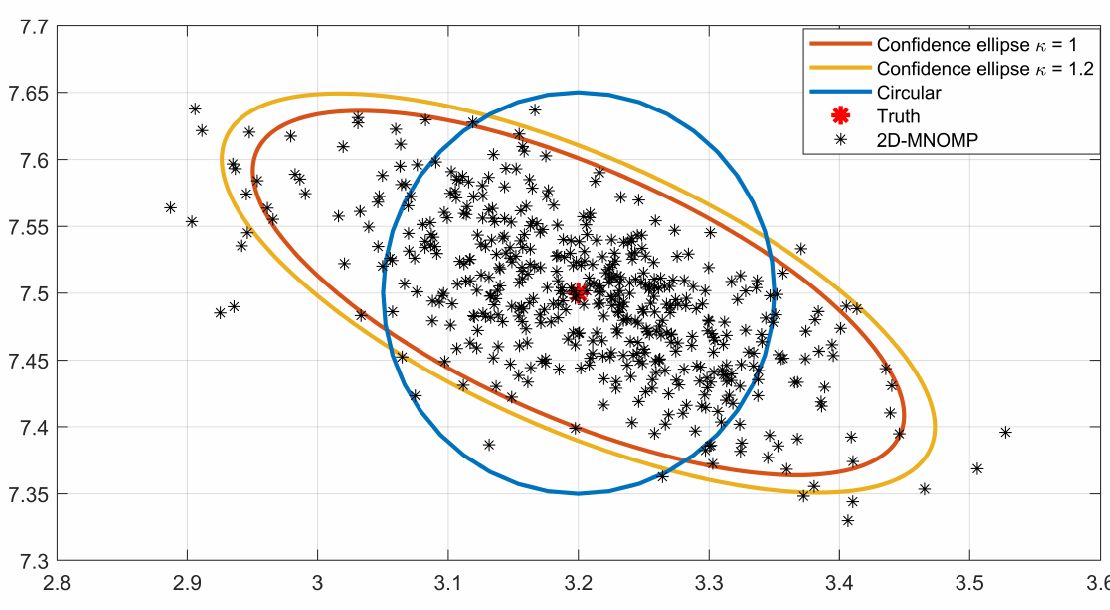}\label{ell2}}
    \caption{The estimates of weak target in multitargets scene output by 2D-MNOMP and the confidence ellipse evaluated by the CRB: (a) The position of the target is ${\mathbf p}=[-2.8,7.5]^{\rm T}$; (b) The position of the target is ${\mathbf p}=[3.2,7.5]^{\rm T}$.} \label{confidence}
\end{figure}

We perform two numerical simulations to validate the good approximation of 2D-MNOMP. The scenario consists of five strong targets with an integrated SNR of 25 dB and one weak target with an integrated SNR of 18 dB. The confidence ellipse is plotted using the integrated SNR of the weak target.
The probability of gate $P_G$ is set as $P_G = 0.95$, where the definition of $P_G$ is below eq. (\ref{2DGate}). After conducting 500 Monte Carlo (MC) trials, we present the measurement results and the corresponding confidence ellipse in Fig. \ref{confidence}, where Fig. \ref{ell1} and Fig. \ref{ell2} depict the confidence ellipses corresponding to the target located at positions $[-2.8,7.5]^{\rm T}$ and $[3.2,7.5]^{\rm T}$, respectively. A simple approach would be to represent the measurement distribution using a circle, where the covariance of the measurements \(\mathbf{z}(t)\) in (\ref{pseudomeasurement}) is modeled as a scaled identity matrix. The radius of the circle is calculated as the average of the major and minor axes of the ellipse. However, as shown in Fig. \ref{confidence}, this simplification does not accurately capture the distribution of the estimates of 2D-MNOMP. In contrast, $94\%$ and $95\%$ of the estimates fall within the confidence ellipse with $\kappa=1$ and $\kappa=1.2$, respectively. Note that the 2D-MNOMP method closely approaches the CRB of a single target, and the CRB information can be incorporated into the data association process during multitarget tracking, thereby enhancing the overall algorithm's performance.

% which demonstrates that the CRB do encapsulates the correlation of the measurements output by 2D-MNOMP.

\subsection{SPA}\label{SPAsubsection}
The measurements in the $t$th frame are described by set $\mathbf{Z}(t) = \left\{ \mathbf{z}_{h}(t) \right\}_{h = 1}^{H(t)}$.
% where the components $\mathbf{z}_{g}^t, g = 1, \cdots, M^t$ have a random order.
The association between measurements and targets can be described by the data association vector $\mathbf{w}(t) = \left[ w_{1}(t), \cdots, w_{K(t)}(t) \right]^{\rm T}$ \cite{WilliamsTAES}, and $w_{k}(t)$ is
\begin{align}
w_{k}(t) \triangleq
\begin{cases}
h \in \left\{1, \cdots, H(t) \right\},      & \text{if at the $t$th frame, target $k$ generates measurement $h$}.\\
0,     & \text{if at the $t$th frame, target $k$ does not generate a measurement.}
\end{cases}
\end{align}
We assume that one target only generates one measurement, which can be described as the following indicator function
\begin{align}\label{constaintorg0}
\psi(\mathbf{w}(t)) \triangleq \begin{cases}
	0, & \text{$\exists k, k' \in \{1, \cdots, K(t)\}$ such that $k \neq k'$ and $w_{k}(t) = w_{k'}(t) \neq 0$}  \\
	1, & \text{otherwise},
\end{cases}
\end{align}
i.e., one measurement can not be originated from the same target. It is hard to obtain the association probability between the measurements and the targets through (\ref{constaintorg0}). Here we  introduce an additional DA vector $\mathbf{b}(t) = \left[ b_{1}(t), \cdots, b_{H(t)}(t) \right]^{\rm T}$ for measurements, i.e.,
\begin{align}
	b_{h}(t) \triangleq
	\begin{cases}
		k \in \left\{1, \cdots, K(t) \right\}, & \text{if at the $t$th frame, measurement $h$ is generated by target $k$}\\
		0, & \text{if at the $t$th frame, measurement $h$ is not generated by a target.}
	\end{cases}
\end{align}
Note that the assumption described by $\psi(\mathbf{w}(t))$ can be expressed by the indicator function
\begin{align}
	\phi \left(\mathbf{w}(t), \mathbf{b}(t) \right) \triangleq \prod_{k = 1}^{K(t)} \prod_{h = 1}^{H(t)} \varphi_{k, h} \left( w_{k}(t), b_h(t) \right),
\end{align}
where
\begin{align}
	\varphi_{k, h} \left( w_{k}(t), b_h(t) \right) \triangleq \begin{cases}
		0, & w_{k}(t) = h, b_h(t) \neq k ~\text{or}~ b_h(t) = k, w_{k}(t) \neq h, \\
		1, & \text{otherwise.}
	\end{cases}
\end{align}
In addition, the predicted state prior PDF $p_{t|t-1}(\mathbf{x}_{k})$ of $\mathbf{x}_{k}$ is
\begin{align}\label{predictedprior}
	p_{t|t-1}(\mathbf{x}_{k}) = \int f(x_{k}(t) \left| x_{k}(t - 1) \right.) p(\mathbf{x}_{k}(t-1)) {\rm d} \mathbf{x}_{k}(t-1),
\end{align}
where $f(x_{k}(t) | x_{k}(t - 1))$ is defined by the transition model given by eq.(\ref{TransitionEquation}) and $p(\mathbf{x}_{k}(t-1))$ is the posterior density of $\mathbf{x}_{k}(t-1)$.

The conditional PMF $q \left({w}_k(t)|\mathbf{x}_k(t) ; \mathbf{Z}(t) \right)$ is given by
\begin{align}
	q \left({w}_k(t)|\mathbf{x}_k(t) ; \mathbf{Z}(t) \right) =\begin{cases}
		{\rm P}_{\rm D}\left(\mathbf{x}_k(t) \right) f\left(\mathbf{z}_{{w}_k(t)} \left| \mathbf{x}_k(t) \right. \right), & {w}_k(t) \in \{1, \cdots, {H}(t)\}, \\
		1 - {\rm P}_{\rm D}\left(\mathbf{x}_k(t) \right), & {w}_k(t) = 0,
	\end{cases}
\end{align}
where ${\rm P}_{\rm D}\left(\mathbf{x}_k(t) \right)$ is the detected probability of the $k$th target and $f\left(\mathbf{z}_{{w}_k(t)} \left| \mathbf{x}_k(t) \right. \right)$ is determined by the measurement model shown in eq.(\ref{measmodel}).
Under the linear Gaussian assumption, the mean value of $f\left(\mathbf{z}_{{w}_k(t)} \left| \mathbf{x}_k(t) \right. \right)$ is known as $\mathbf{H} \mathbf{x}_k(t)$ and its covariance can be calculated as
\begin{align}
	\mathbf{S}_k(t | t - 1) = \mathbf{H} \boldsymbol{\Sigma}_k(t | t - 1) \mathbf{H}^{\rm T} + \mathbf{R}_k(t).
\end{align}

The message $\beta_{[0] k \rightarrow {w}_k } \left( t \right)$ transmitted from the factor node $q \left({w}_k(t)|\mathbf{x}_k(t) ; \mathbf{Z}(t) \right)$ to the variable node ${w}_k(t)$ is
\begin{align}
	\beta_{[0] k \rightarrow {w}_k } \left( t \right) = \int q \left(\mathbf{x}_k(t), {w}_k(t) ; \mathbf{Z}(t) \right) p_{t|t-1}(\mathbf{x}_{k}) \mathrm{d} \mathbf{x}_{k}(t).
\end{align}
Similarly, the distribution of $\mathbf{b}(t)$ is 
\begin{align}\label{}
	h\left(b_{h}(t); \mathbf{z}_{h}(t) \right) = \begin{cases}
		1, & b_{h}(t) \in \left\{ 1, \ldots, K(t) \right\}, \\
		\mu_c f_c\left( \mathbf{z}_{h}(t) \right), & b_{h}(t) = 0,
	\end{cases}
\end{align}
where $f_c\left( \mathbf{z}_{h}(t) \right)$ is the location distribution function of clutter. We let $\xi_{[0] h \rightarrow b_{h}}(t) = h\left(b_{h}(t); \mathbf{z}_{h}(t) \right)$ be the initialized association weights. In summary, the factor graph \cite{SPM, FCref} is shown in Fig. \ref{FactorGraph}.

\begin{figure}
    \centering
    \includegraphics[width = 120mm]{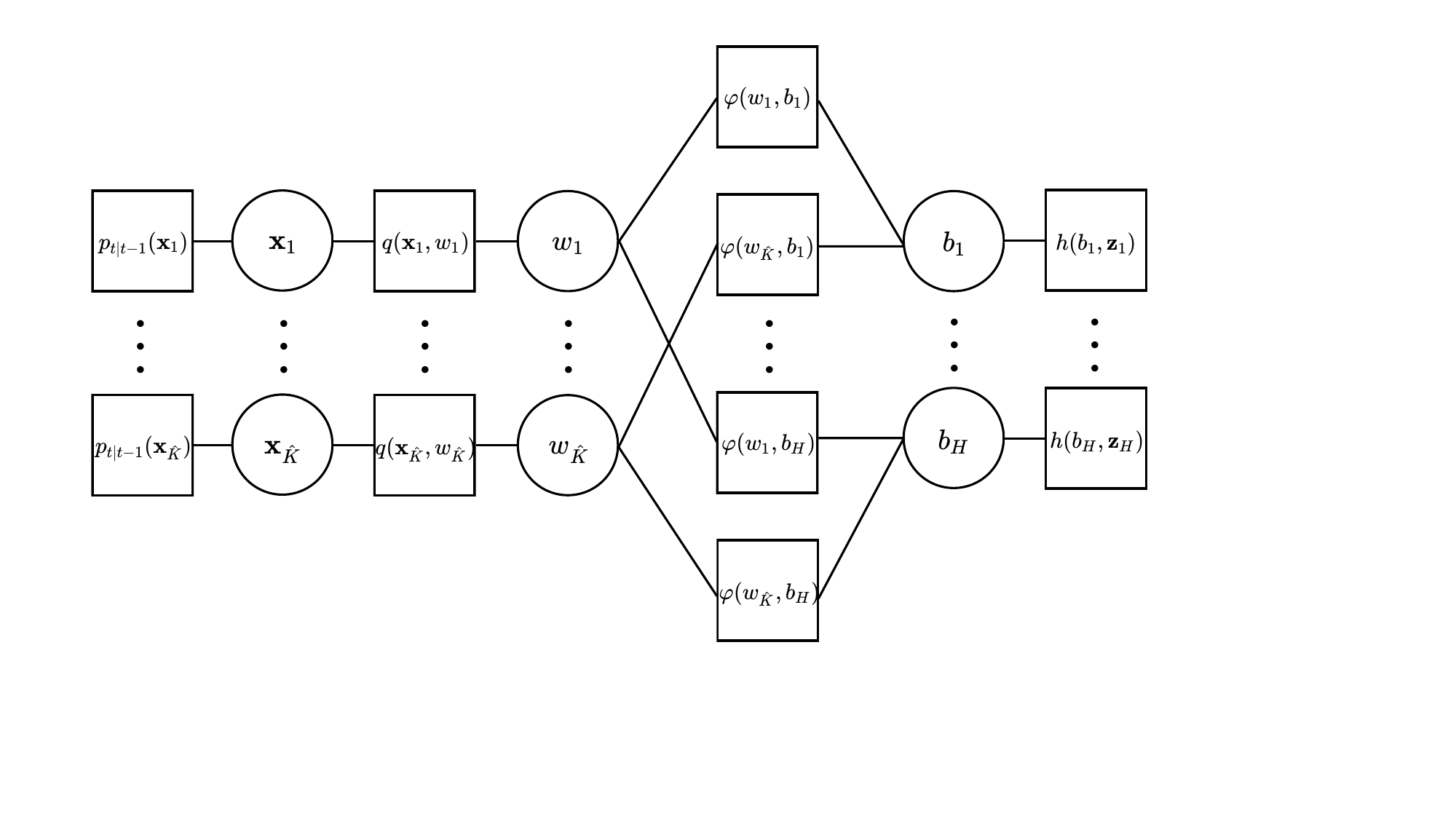}
    \caption{Factor graph for multi-targets data association in a single mmWave radar.} \label{FactorGraph}
\end{figure}

% With the help of predicted state prior PDF, the probabilistic mass distribution of data association vector $\mathbf{w}(t)$ can be initialized as

% where
% To apply the SPA algorithm, we plot the the factor figure in Fig. \ref{FactorGraph} stretch the factor node $\psi(\mathbf{w}(t))$ by introducing the DA vector

% The goal of probabilistic data association is to calculate the posterior data association probability mass functions $f(w_k(t) | \mathbf{Z}(t))$.
% According to the factor graph, the message passed from the factor nodes
% Similarly, the messages passed from the factor nodes $q \left(\mathbf{x}_k(t), {w}_k(t) ; \mathbf{Z}(t) \right)$ to the variable nodes ${w}_k(t)$ are given by

In order to decrease the complexity, we only consider the measurements which belong to the validation region calculated by the target prior distribution $p_{t|t-1}(\mathbf{x}_{k})$ (\ref{predictedprior}) and $P_G$.
In other words, the valid measurement indexes obey $h \in \mathbb{S}_{k}(t)$, where the valid measurement set $\mathbb{S}_{k}(t)$ is 
\begin{align}\label{ValidSet}
	{\mathbb S}_{k}(t) = \left\{ h \left| \left( \left[ p_{x,h}, p_{y,h}, v_h \right]^{\rm T} - \mathbf{H}_{k}^{\prime} {\mathbf x}_k(t | t - 1) \right)^{\rm T} \left( \mathbf{S}_{k}^{\prime} \right)^{-1} \left( \left[ p_{x,h}, p_{y,h}, v_h \right]^{\rm T} - \mathbf{H}_{k}^{\prime} {\mathbf x}_k(t | t - 1) \right) \le d_G \right. \right\},
\end{align}
where
\begin{subequations}
	\begin{align}
	\mathbf{H}_{k}^{\prime} &=
	\begin{bmatrix}
	1 & 0 & 0 & 0 \\
	0 & 0 & 1 & 0 \\
	0 & \sin \theta_{k} & 0 & \cos \theta_{k}
	\end{bmatrix}, \\
    \mathbf{S}_k^{\prime} &= \mathbf{H}_{k}^{\prime} \boldsymbol{\Sigma}_k(t | t - 1) \mathbf{H}_{k}^{\prime \rm T} + \mathbf{R}_k^{\prime}(t), \label{covarofSkp} \\
	\mathbf{R}_k^{\prime}(t) &= \kappa{\rm CRB}([p_x, p_y, v]^{\rm T})\big|_{[p_x, p_y, v]=[\hat{p}_x, \hat{p}_y, \hat{v}]},
\end{align}
\end{subequations}
${\rm CRB}([p_x, p_y, v]^{\rm T})$ is given by (\ref{CRBpxpyv}), and $d_G = F_{\chi^2(3)}^{-1}(P_G)$. The deduction details of ${\mathbb S}_k(t)$ and $d_{G}$ are shown in Appendix \ref{DofVR}.
The initialized association of invalid targets will be set as $0$. Traditional SPA employ 2D-Gate to perform data association, as specified in eq. (\ref{2DGate}) in Appendix \ref{DofVR}. Here we employ the radial velocity information to perform data association, which suppresses the false alarms in the real experiments shown later.
% Like validation region, we also make the verification with the radial velocity of measurement and the velocity of target.
% For the $h$th measurement and $k$th target, the velocity verification is
% \begin{align}
% \left| {{v}_{x,k}}(t)\sin \theta_{k}(t) +{{v}_{y,k}}(t)\cos \theta_{k}(t) - {{v}_{h}}(t) \right|<{{v}_{\text{Th}}}
% \end{align}
% where $v_{\rm Th}$ is the threshold of velocity.
% For the association which doesn't pass the velocity verification, the initialized association $\beta_{[0] k \rightarrow h} = 0$.

% If $\mathbf{z}_{h}(t)$ is out of the validation region of target ${\mathbf x}_k(t | t - 1)$, the corresponding initialized probability $\beta_{[0] k \rightarrow h}(t)$ will be set as zero.
% And validation region ${\mathbb S}_k(t)$ can be written as
% \begin{align}\label{ValidationRegion}
% {\mathbb S}_{k}(t) = \left\{ {\mathbf{z}_{h}(t)} \left| \left( {\mathbf{z}_{h}(t)} - \mathbf{H} \mathbf{x}_k(t | t - 1) \right)^{\rm T} {\mathbf S}_k^{-1}(t | t - 1) \left( {\mathbf{z}_{h}(t)} - \mathbf{H} \mathbf{x}_k(t | t - 1) \right) \right. \leq d_{G}^2 \right\},
% \end{align}

% And the messages passed from the factor nodes $u \left(b_{h}(t) ; \mathbf{z}_{h}(t) \right)$ to the variable nodes $b_{h}(t)$ are given by $\xi_{h}^{t}\left(b_{h}(t) \right) = h\left(b_{h}(t) ; \mathbf{z}_{h}(t) \right)$, where

The message passed between nodes $w_k(t)$ and $b_h(t)$ can be shown as
\begin{subequations}
\begin{align}\label{}
&\delta_{[i](k \rightarrow h)} = \frac{\beta_{[0] k \rightarrow h}(t)}{\beta_{ [0] k \rightarrow 0}(t) + \sum_{{h^{\prime} \neq h} } \beta_{[0] k \rightarrow h^{\prime}}(t) v_{[i]\left(h^{\prime} \rightarrow k\right)}} \\
&v_{[i](h \rightarrow k)} =  \frac{ \xi_{[0]h \rightarrow k}(t) }{\xi_{[0] h \rightarrow 0}(t) + \sum_{k^{\prime} \neq k} \xi_{[0] h \rightarrow k^{\prime}}(t) \delta_{[i - 1]\left(k^{\prime} \rightarrow h\right)}}
\end{align}
\end{subequations}
for $k = 1, \cdots, K(t)$ and $h = 1, \cdots, H(t)$, $i$ is the index of iteration. For this iterative algorithm we initialize $\delta_{[0]\left(k \rightarrow h \right)}$ as $\delta_{[0]\left(k \rightarrow h \right)} = { \beta_{k \rightarrow h}(t) } / {\beta_{k \rightarrow 0}(t)}$.
And the association probabilities of $\mathbf{w}_{k}^t$ and $\mathbf{b}_{h}^t$ can be shown as
\begin{subequations}
\begin{align}
	\beta_{k \rightarrow h}(t) &= \frac{ \beta_{[0] k \rightarrow h}(t) v_{[i]\left(h \rightarrow k\right)}}{\beta_{[0] k \rightarrow 0}(t) + \sum_{{h^{\prime} = 1} }^{H(t)} \beta_{[0] k \rightarrow h^{\prime}}(t) v_{[i]\left(h^{\prime} \rightarrow k\right)}}, \\
	\xi_{h \rightarrow k}(t) &= \frac{\xi_{[0] h \rightarrow k}(t) \delta_{[i]\left(k \rightarrow h\right)}}{\xi_{[0] h \rightarrow 0}(t) + \sum_{k^{\prime} = 1}^{K(t)} \xi_{[0] h \rightarrow k^{\prime}}(t) \delta_{[i]\left(k^{\prime} \rightarrow h\right)}}.
\end{align}
\end{subequations}
After obtaining the association probability of each target, we can calculate the posterior distribution of each target at frame $t$ by using the posterior association probabilities.

\subsection{Kalman Filter}\label{KF}
% The association between measurements and targets can be described by the data association vector $\mathbf{w}(t) = \left[ w_{1}(t), \ldots, w_{K}(t) \right]^{\rm T}$, whose $k$th component $w_{k}(t)$ satisfies
% \begin{align}
% 	w_{k}(t)=
% 	\begin{cases}
% 		h \in \left\{1, \cdots, H(t) \right\},      & \text{if target $k$ generates measurement $h$},\\
% 		0,     & \text{if target $k$ does not generate a measurement.}
% 	\end{cases}
% \end{align}
% Define the probability $\beta_{k \rightarrow h}(t)$ of the event that measurement $h$ belongs to target $k$, i.e.,
% \begin{align}
% 	\beta_{k \rightarrow h}(t) = p\left(w_{k}(t) = h \left| \mathbf{Z}(1 : t) \right. \right),
% \end{align}
% where $\mathbf{Z}(1 : t)\triangleq\{\mathbf{Z}(1),\cdots,\mathbf{Z}(t)\}$ is all the measurements up to time stamp $t$. Obviously, $\beta_{k \rightarrow 0}(t)$ denotes the probability that target $k$ does not generate a measurement and satisfies
% \begin{align}
% 	\beta_{k \rightarrow 0}(t) = 1 - \sum_{h = 1}^{H} \beta_{k \rightarrow h}(t).
% \end{align}
% i.e., we have obtained the association probabilities of all the targets and measurements.
After the targets and measurements have been associated properly by SPA module shown in Subsection \ref{SPAsubsection}, target tracking via the KF can be implemented as follows \cite{EstAppTraNav, PDAfilter}:
\begin{itemize}
	\item Prediction and Minimum Prediction MSE Matrix Step:
	\begin{subequations}\label{StateTransition}
		\begin{align}
			\hat{\mathbf{x}}_{k}(t | t - 1) &= \mathbf{A} \hat{\mathbf{x}}_{k}(t - 1) \\
			\boldsymbol{\Sigma}_{k}(t | t - 1) &= \mathbf{A} \boldsymbol{\Sigma}_{k}(t - 1) \mathbf{A}^{\rm T} + \boldsymbol{\Gamma} {\mathbf Q}_k(t - 1) \boldsymbol{\Gamma}^{\rm T}
		\end{align}
	\end{subequations}
	\item Kalman Gain Matrix Step:
	\begin{align}
		\mathbf{K}_{k}(t) = \boldsymbol{\Sigma}_{k}(t | t - 1) \mathbf{H}^{\rm T} \left(\mathbf{R}_{k}(t) + \mathbf{H} \boldsymbol{\Sigma}_{k}(t | t - 1) \mathbf{H}^{\rm T}\right)^{-1}.
	\end{align}
	\item Correction and Minimum MSE Matrix Step:
	\begin{subequations}\label{KalmanFilter}
		\begin{align}
			\hat{\mathbf{x}}_{k}(t) &= \hat{\mathbf{x}}_{k}(t | t - 1) + \mathbf{K}_{k}(t) \bar{\boldsymbol{\delta}}_{k}(t) \\
			\boldsymbol{\Sigma}_{k}(t) &= \boldsymbol{\Sigma}_{k}(t | t - 1) - \left[1 - \beta_{k \rightarrow 0}(t) \right] \mathbf{K}_{k}(t) \mathbf{H} \boldsymbol{\Sigma}_{k}(t | t - 1) + \tilde{\boldsymbol{\Sigma}}_{k}(t).
		\end{align}
	\end{subequations}
	where $\bar{\boldsymbol{\delta}}_{k}(t)$ is
	\begin{align}
		\bar{\boldsymbol{\delta}}_{k}(t) = \sum_{h = 1}^{H(t)} \beta_{k \rightarrow h}(t){\mathbf e}_{h,k}(t),
	\end{align}
	the covariance matrix $\tilde{\boldsymbol{\Sigma}}_{k}(t)$ is
	\begin{align}
		\tilde{\boldsymbol{\Sigma}}_{k}(t) = \mathbf{K}_{k}(t) \left[ \sum_{h = 1}^{H(t)} \beta_{k \rightarrow h}(t) {\mathbf e}_{h,k}(t) {\mathbf e}_{h,k}^{\rm T}(t) - \bar{\boldsymbol{\delta}}_{k}(t) \bar{\boldsymbol{\delta}}_{k}^{\rm T}(t) \right] \mathbf{K}_{k}^{\rm T}(t),
	\end{align}
	and ${\mathbf e}_{h,k}(t)$ denotes the predicted $h$th measurement error using the $k$th target state given as
	\begin{align}
		{\mathbf e}_{h,k}(t)=\mathbf{z}_{{h}}(t) - \mathbf{H} \hat{\mathbf{x}}_{k}(t | t - 1).
	\end{align}
\end{itemize}

%It's important to note that although we have outlined the standard Kalman filter steps for multiple target tracking, several questions naturally arise. Firstly, the measurements in equation (\ref{meas}) correspond to the positions of multiple targets along the $x$ and $y$ axes, and these need to be extracted in real-time from the baseband measurement model (\ref{radarmeas}). Additionally, the covariance matrix $\mathbf{R}(t)$ in the measurement model (\ref{measmodel}) is unknown and requires estimation. Secondly, efficient data association is crucial for continuous target tracking. This involves associating measurements with existing targets and establishing new tracks as needed. Thirdly, radial velocities of multiple targets are simultaneously estimated from the baseband data. These radial velocities should also be utilized for data association since the radial velocity is closely related to the state of the target. In the subsequent sections, we will address and resolve these challenges step by step.
\subsection{Appearance and Disappearance of Targets}
After obtaining the association probability for existed targets and measurements, we calculate the posterior distribution of targets via the KF shown in eq.(\ref{KalmanFilter}).
And we will discuss the appearance and disappearance of targets in this subsection.
If a measurement $\mathbf{z}_h(t)$ whose association probability with no existed target $\tilde{p}({b}_{h}(t) = 0 )$ is larger than $0.5$, we believe this measurement is generated by a new target.
The new target will be created based on the measurement, the position and velocity of the target will be set as
\begin{align}
    \mathbf{x}_{\rm new}(t) = \left[ r_h(t) \sin \theta_h(t), v_h(t) \sin \theta_h(t), r_h(t) \cos \theta_h(t), v_h(t) \cos \theta_h(t) \right]^{\rm T}
\end{align}
where $r_h(t), v_h(t), \theta_h(t)$ are the radial range, radial velocity and azimuth angle of the measurement given by 2D-MNOMP-CFAR method directly. We then perform DA and KF for the new targets at the next frame.

Conversely, if a target $\mathbf{x}_k(t)$ whose association probability with no measurements $\tilde{p}({w}_{k}(t) = 0)$ is larger than $0.5$, we believe that the target is missed detected at $t$th frame. We use a hyperparameter $N_{\rm ext}$ which denotes the number of extrapolation. Provided that the target is miss detected, we extrapolate its state using the previous results. If the target is missed detected for  consecutive $N_{\rm ext}$ times, we believe the target disappeared. 

% After a target disappears, we will judge whether the target is a false alarm or not.

%The criterion of judgement is whether the valid length of trajectory (cut the $N_{\rm ext}$ points at the trajectory) is great than the set threshold.
%Therefore, we design the sceond hyperparameter to denote the threshold of the length for a valid target $N_{\rm val}$.
%As a result, the MNOMP-SPA gives the history of valid targets.

\section{numerical results}
In this section, we verify the theoretical results and evaluate the performance of the proposed MNOMP-SPA-KF by comparing with other methods.

% Firstly, we will show the MNOMP-SPA-KF algorithm results in the multi-target scenario intuitively.
% Secondly, the effectiveness of super-resolution LSE\&D algorithm is shown by the comparison of MNOMP-SPA-KF and FFT-SPA-KF.
% Thirdly, we will compare the performance of SFNC-SPA algorithm under the condition with and without clustering algorithm.
% And finally, we will show the performance under the conditions of different ${\rm SNR}$s.

\subsection{Scenario Generation, Tracking Metrics and Benchmark Algorithms}
We consider the region of interest (ROI) given by $[- 30m, 30m] \times [0, 60m]$, with up to $20$ targets.
The total observation time is typically set as $T_{\rm tol} = 6$s and the observation interval is $0.1$s, which corresponds to $T_{\rm sam} = 60$ frames.
The initial positions are generated randomly within the range of $[- 10m, 10m] \times [6m, 24m]$; and the velocities are generate randomly within the range of $[- 3m/s, 3m/s] \times [-1m/s, 5m/s]$.
In that case, the possible appearance range of targets can be calculated as $[- 28m, 28m] \times [0, 60m]$.
The noise covariance matrix of $\mathbf{w}$ in eq.(\ref{TransitionEquation}) is set as ${\rm diag}\left\{ [10^{-6}, 10^{-6}]^{\rm T} \right\}$.
The clutter number is Possion distributed with mean $\mu_{\rm c} = 4$.

The parameters of mmWave radar are set as follows: Number of fast time samples, number of pulses in one CPI and number of receivers are $N = 128$, $M = 64$ and $L = 8$, respectively.
The idle time is $T_{\rm idle} = 100 \mu s$, the ramp time is $T_{\rm ramp} = 60 \mu s$, corresponding to circle time $T_{\rm c} = 160 \mu s$.
The ranges of the radial distance, the velocity and the azimuth are $[0, 100.0]m$, $[- 6.08, 6.08] m/s$ and $[- \pi / 2, \pi / 2]$, respectively.
The range resolution and velocity resolution are $0.78$m and $0.19 {\rm m} / {\rm s}$, respectively. The integrated ${\rm SNR}$ is set as $19 {\rm dB}$, which are constant for all the targets at each moment.

The parameters of the MNOMP-SPA-KF are set as follows: False alarm rate ${P}_{\rm FA} = 0.01/(128\times 64)$ \cite{FundamentalsRadarSP, Kaydet}; the training band and guard band is $[5, 4]$ and $[3, 3]$, which means the number of training cells is $50$; the multiple factor can be calculated as $4.2257$, i.e., $6.2590$dB; the upper bound of the number of targets is $K_{\rm max} = 30$; the gate probability is $P_{\rm G} = 0.95$; the number of iterations of SPA is $10$. The multiple factor of CRB is $\kappa = 1.2$, which is used to obtain the measurement covariance matrix ${\mathbf R}_k(t)$ of the pseudo noise ${\mathbf v}(t)$ (\ref{pseudomeasurementmnomp}).

The time-average mean optimal subpattern assignment (MOSPA) \cite{OSPA} is used to evaluate the performance of tracking algorithm, which is averaged over $300$ MC trials unless stated otherwise. The OSPA distance is calculated between object set $\mathfrak{X}$ and track set $\mathfrak{Y}$, which are
\begin{align}
\mathfrak{X}_k & =\left\{\left(\ell_1, \mathbf{x}_{k, 1}\right), \ldots,\left(\ell_m, \mathbf{x}_{k, m}\right)\right\}, \\
\mathfrak{Y}_k & =\left\{\left(s_1, \mathbf{y}_{k, 1}\right), \ldots,\left(s_n, \mathbf{y}_{k, n}\right)\right\},
\end{align}
where $\ell, s$ are the track label, $\mathbf{x}, \mathbf{y}$ are the state vectors and $m, n$ are the number of real tracks and estimated tracks.
For $m < n$, the OSPA distance between $\mathfrak{X}$ and $\mathfrak{Y}$ can be shown as \cite{OSPA}
\begin{align}
D_{p, c}\left(\mathfrak{X}_k, \mathfrak{Y}_k\right) =\left[\frac{1}{n}\left(\min _{\pi \in \Pi_n} \sum_{i=1}^m\left(d_c\left(\tilde{\mathbf{x}}_{k, i}, \tilde{\mathbf{y}}_{k, \pi(i)}\right)\right)^p+(n-m) \cdot c^p\right)\right]^{\frac{1}{p}}
\end{align}
where $\tilde{\mathbf{x}}_{k, i} \equiv\left(\ell_i, \mathbf{x}_{k, i}\right), \tilde{\mathbf{y}}_{k, \pi(i)} \equiv\left(s_{\pi(i)}, \mathbf{y}_{k, \pi(i)}\right)$ and $d_c(\tilde{\mathbf{x}}, \tilde{\mathbf{y}})=\min (c, d(\tilde{\mathbf{x}}, \tilde{\mathbf{y}}))$ is the cutoff distance between two tracks at $t_k$, with $c>0$ being the cutoff parameter; $d(\tilde{\mathbf{x}}, \tilde{\mathbf{y}})$ is the base distance between two tracks at $t_k$; $\Pi_n$ represents the set of permutations of length $m$ with elements taken from $\{1,2, \ldots, n\}$; $1 \leq p<\infty$ is the OSPA metric order parameter.

% To show the performance of the MNOMP-SPA-KF which integrates the super resolution MNOMP estimation and detection, the data association SPA, the KF 
Now we construct the benchmark algorithms by analyzing the modules in completing the tracking task. Note that the proposed MNOMP-SPA-KF consists of three modules, namely, the estimation and detection module using MNOMP, the data association module using SPA, the KF module implemented by utilizing CRB to characterize the correlation of the measurements. A naive solution is to replace those three modules with the existing methods and integrates them together to perform target detection and tracking. For target estimation and detection, we have implemented 
\begin{itemize}
    \item 2D-FFT-CFAR \cite{FundamentalsRadarSP}: We perform two dimensional FFT on the fast and slow time domain to obtain the range-Doppler heat map for each antenna, which is then averaged over the spatial domain and input to the two dimensional CFAR detector.
    \item 2D-MNOMP-CFAR: The implementation details have been described in Section \ref{2DMNOMP}.
\end{itemize}
For data association, we have implemented
\begin{itemize}
    \item PDA \cite{PDA}: The PDA algorithm calculates the association probabilities to the target being tracked for each validated measurement at the current time. The PDA algorithm performs well in the single target scene.
    \item JPDA \cite{JPDA}: The JPDA approach is the extension of the PDA. The key feature of the JPDA is that it treats the data association process probabilistically and considers all possible measurement-to-target assignments. The JPDA is designed to address the multi-targets scene. In the following experiments, the JPDA method is run by the MATLAB function ``trackerJPDA'' \cite{trackerJPDA}.
    % \item SPA with 2D gate \cite{}: This is a very standard sum product algorithm (SPA) employing the position information for data association.
    \item SPA with 3D-Gate: This incorporates the radial velocity information to perform data association, which is an extension of the SPA with 2D-Gate. The implementation details have been described in Section \ref{SPAsubsection}\footnote{SPA with 2D-Gate is a very standard SPA employing the positional information for data association.}.
\end{itemize}
For target tracking, we have implemented
\begin{itemize}
    \item subKF: The subKF is implemented by ignoring the correlation of the measurements, i.e., the covariance of the measurements ${\mathbf z}(t)$ in (\ref{pseudomeasurement}) is a scaled identity matrix. 
    \item KF: The Kalman filter is implemented by taking the correlation of the measurements into account, i.e., the covariance of the measurements ${\mathbf z}(t)$ in (\ref{pseudomeasurement}) is a scaled CRB, whose details are shown in Section \ref{KF}.
\end{itemize}
Based on the above methods, a natural idea is to cascade the methods from different stages. Because the MATLAB built-in function ``trackerJPDA'' incorporates JPDA and subKF, we only run JPDA-subKF. Consequently, we have implemented $10$ methods, i.e., FFT-PDA-subKF, FFT-PDA-KF, FFT-JPDA-subKF, FFT-SPA-subKF, FFT-SPA-KF, MNOMP-PDA-subKF, MNOMP-PDA-KF, MNOMP-JPDA-subKF, MNOMP-SPA-subKF, MNOMP-SPA-KF (3D-Gate).
To show the benefit of adopting the 3D-Gate in data association, MNOMP-SPA-KF (2D-Gate) is also implemented, where the SPA (2D-Gate) only exploits the two dimensional positional information.
In the following, we evaluate these methods and demonstrate the advantage of the proposed integration, i.e., MNOMP-SPA-KF (3D-Gate).

\subsection{Verification of CRB}
In this subsection, we validate the correctness of the ${\rm CRB}\left[ p_x, p_y \right]$ (\ref{CRBpxpy}). We set the number of targets as $1$ and we use the 3D-NOMP and 2D-MMNOMP algorithm to estimate the targets' radial range and azimuth angle, then we obtain the position estimates. The results are averaged over 300 MC trials, and results are shown in Fig. \ref{MSEandCRB}.

\begin{figure}
	\centering
	\subfigure[]{
	\label{MSEofpx}
	\includegraphics[width = 60mm]{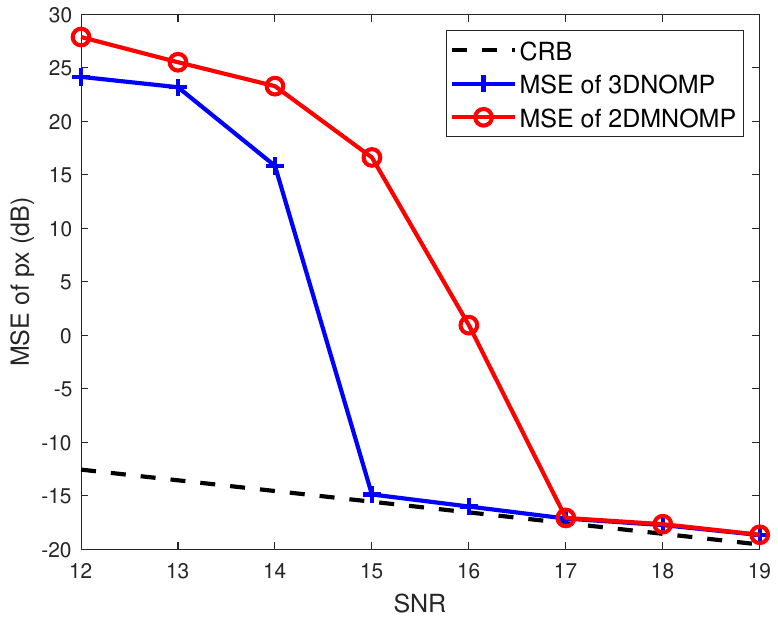}}
	\subfigure[]{
	\label{MSEofpy}
	\includegraphics[width = 60mm]{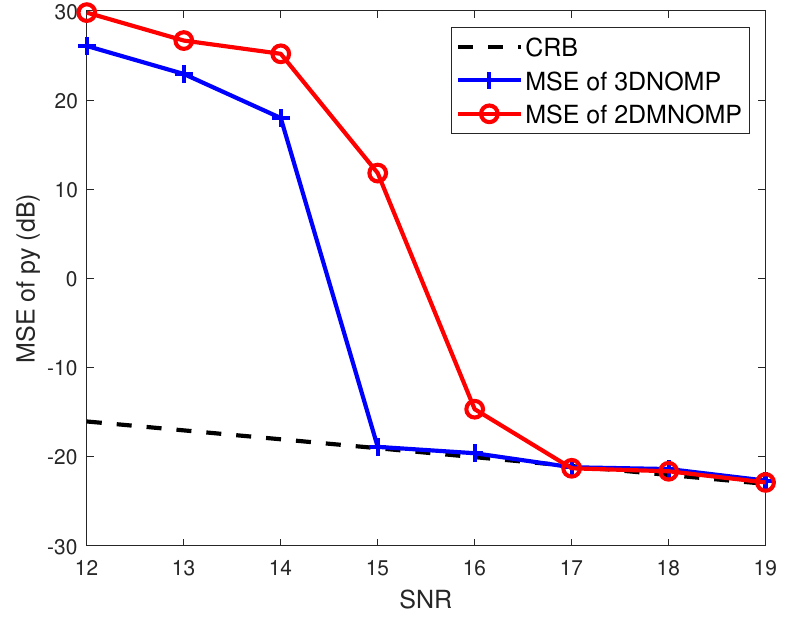}}
	\caption{The CRB in eq.(\ref{CRBpxpy}) and the estimated MSE by 3D-NOMP and 2D-MNOMP. (a) MSE and CRB of $p_x$. (b) MSE and CRB of $p_y$.} \label{MSEandCRB}
\end{figure}

From Fig. \ref{MSEandCRB}, we observe that there exists a threshold phenomenon: When the SNR exceeds $15$ dB, the MSE of the 3D-NOMP closely matches to the CRB, demonstrating that the 3D-NOMP asymptotically approaches to the CRB as SNR increases; while when the SNR exceeds $17$ dB, the 2DMNOMP performs well and approaches to 3D-NOMP and CRB.

\subsection{Performance Comparisons}
\begin{figure}
    \centering
    \subfigure[\scriptsize{FFT-PDA-subKF}]{
    \label{FFT_PDA_KF_sub_6tar}
    \includegraphics[width = 36mm]  {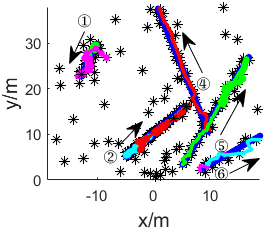}}
    \subfigure[\scriptsize FFT-PDA-KF]{
    \label{FFT_PDA_KF_6tar}
    \includegraphics[width = 36mm]{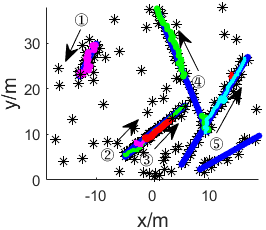}}
    \subfigure[\scriptsize FFT-JPDA-subKF]{
    \label{FFT_JPDA_KF_6tar}
    \includegraphics[width = 36mm]{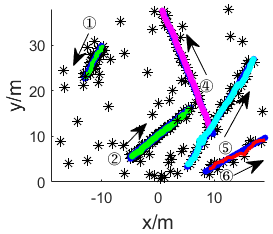}}
    \subfigure[\scriptsize FFT-SPA-subKF]{
    \label{FFT_SPA_subKF_6tar}
    \includegraphics[width = 36mm]{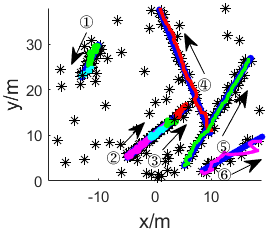}}
    \subfigure[\scriptsize FFT-SPA-KF]{
    \label{FFT_SPA_KF_6tar}
    \includegraphics[width = 36mm]{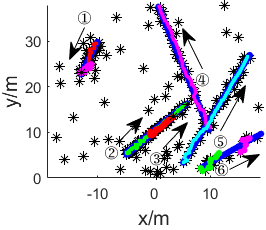}}
    \subfigure[\scriptsize MNOMP-PDA-subKF]{
    \label{MNOMP_PDA_subKF_6tar}
    \includegraphics[width = 36mm]{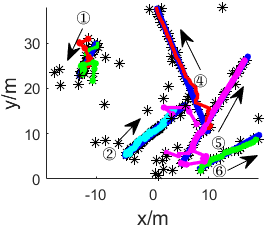}}
    \subfigure[\scriptsize MNOMP-PDA-KF]{
    \label{MNOMP_PDA_KF_6tar}
    \includegraphics[width = 36mm]{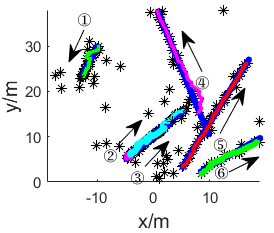}}
    \subfigure[\scriptsize MNOMP-JPDA-subKF]{
    \label{MNOMP_JPDA_KF_6tar}
    \includegraphics[width = 36mm]{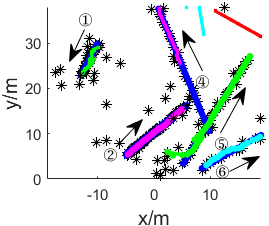}}
    \subfigure[\scriptsize MNOMP-SPA-subKF]{
    \label{MNOMP_SPA_subKF_6tar}
    \includegraphics[width = 36mm]{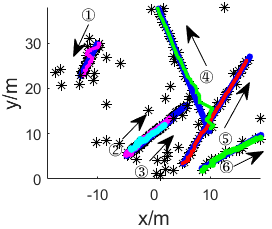}}
    \subfigure[\scriptsize 3DNOMP-SPA-KF]{
    \label{3DNOMP_SPA_KF_6tar}
    \includegraphics[width = 36mm]{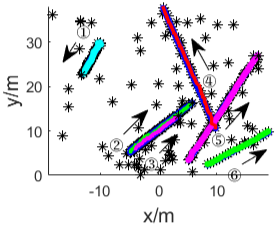}}
    \subfigure[\tiny MNOMP-SPA-KF (2D-Gate)]{
    \label{MNOMP_SPA(2D_gate)_KF_6tar}
    \includegraphics[width = 36mm]{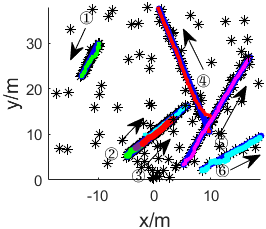}}
    \subfigure[\scriptsize MNOMP-SPA-KF]{
    \label{MNOMP_SPA_KF_6tar}
    \includegraphics[width = 36mm]{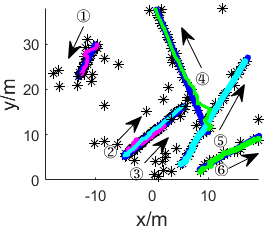}}
	\caption{The tracking results of different algorithms. The arrow indicates the direction of the target movement.}\label{Exp4scene}
\end{figure}
\begin{figure}
	\centering
	\includegraphics[width = 150mm]{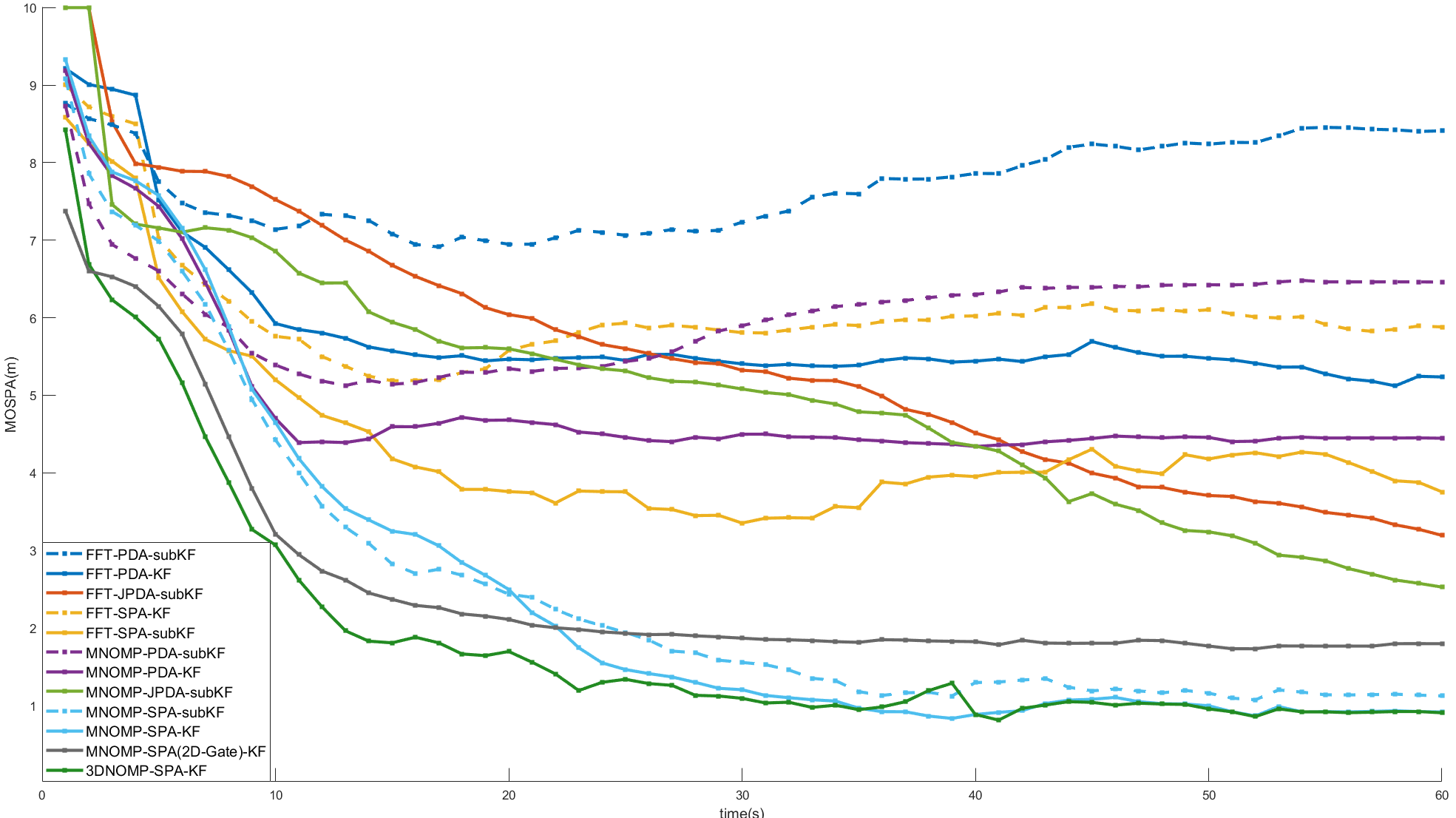}
	\caption{The time-average MOSPA of different algorithms.}\label{MOSPAcompare_all}
\end{figure}

There exist six targets, where target \textcircled{2} and target \textcircled{3} are very close at the beginning, i.e., they are only 0.5 m apart and their velocity difference is 0.4 m/s, compared to the range resolution $0.78$ m and velocity resolution $0.19$ m/s. Here, 3DNOMP-SPA-KF is also implemented where the radial distances, radial velocities and the azimuths of multiple targets are extracted by performing the $3$ dimensional LSE, although it is time consuming. And for FFT-based algorithms, we applied the DBSCAN clustering algorithm to group detections from the same target and suppress clutters \cite{DBSCAN}.

Fig. \ref{Exp4scene} illustrates the target trajectories produced by the conventional algorithms, the proposed 2DMNOMP-SPA-KF (2D-Gate), 2DMNOMP-SPA-KF (3D-Gate) and 3DNOMP-SPA-KF (3D-Gate).
  % Specifically, the conventional estimation and detection algorithm used is FFT-CFAR, while the data association algorithms are PDA and JPDA. Additionally, the suboptimal Kalman Filter (subKF) assumes that the measurements are uncorrelated. 
  % Consequently, we can construct the following comparative algorithms: FFT-SPA-KF, FFT-SPA-subKF, FFT-PDA-KF, FFT-PDA-subKF, FFT-JPDA-KF, MNOMP-PDA-KF, MNOMP-PDA-subKF, MNOMP-JPDA-KF, MNOMP-SPA-subKF, and MNOMP-SPA-subKF. The JPDA-KF based method is realized by matlab function "trackerJPDA".
As shown in Fig. \ref{FFT_PDA_KF_sub_6tar} $\sim$ Fig. \ref{FFT_SPA_KF_6tar}, the FFT-based method struggles to distinguish nearby targets, such as target \textcircled{2} and target \textcircled{3}, in a timely manner. It can only reliably separate their trajectories when the targets are sufficiently spaced apart (approximately 4 meters). For subKF based methods, Fig. \ref{MNOMP_PDA_subKF_6tar} shows that due to the strong interference from clutter, MNOMP-PDA-subKF generates two trajectories around target \textcircled{5},
which nearly overlap as time goes on.
Fig. \ref{MNOMP_JPDA_KF_6tar} shows that MNOMP-JPDA-subKF fails to track target \textcircled{4} at the beginning. Fig. \ref{MNOMP_SPA_subKF_6tar} shows that MNOMP-SPA-subKF successfully tracks all the targets at the beginning, but fails to track target \textcircled{2} and target \textcircled{3} at the end. Although JPDA-subKF also demonstrates high association accuracy, its computational complexity is higher than that of SPA-KF \footnote{Numerically it if found that JPDA is an order of magnitude slower than SPA.}.
  
  As for NOMP and KF based methods,  Fig. \ref{MNOMP_PDA_KF_6tar} shows that MNOMP-PDA-KF is more stable than the previously described methods, but fails to track target \textcircled{4} at the beginning.
  Fig. \ref{MNOMP_SPA(2D_gate)_KF_6tar} shows that MNOMP-SPA-KF (2D-Gate) fails to track target \textcircled{3} at the beginning.
  As shown in Fig. \ref{3DNOMP_SPA_KF_6tar} and Fig. \ref{MNOMP_SPA_KF_6tar}, the proposed 3DNOMP-SPA-KF and MNOMP-SPA-KF methods generate the most stable trajectories, indicating that they are less affected by clutter. Unlike the previously described methods, the 3DNOMP-SPA-KF and MNOMP-SPA-KF are capable of accurately tracking the targets all the time.
  % FFT-based methods, the 3DNOMP-SPA-KF and MNOMP-SPA-KF are capable of accurately tracking the trajectories of nearby targets from the outset.}
  
We perform extensive numerical simulations to highlight the advantages of the proposed MNOMP-SPA-KF and 3DNOMP-SPA-KF. After conducting 300 MC trials, we calculate the MOSPA at each time step and present the results in Fig. \ref{MOSPAcompare_all}. The MOSPAs of MNOMP-SPA-KF and 3DNOMP-SPA-KF are a little lower than MNOMP-SPA-subKF and significantly lower than the other methods, demonstrating the effectiveness of the integration.

\subsection{Performance of MNOMP-SPA-KF in Terms of Detecting the Weak Target}
% \subsection{Performance Comparison with the Traditional CFAR detection in Terms of Detecting the weak Target}

In this section, the weak target detection ability of MNOMP-SPA-KF is investigated. There exist three targets, two of which are strong with time domain SNRs being $-23.9$ dB, the other is a weak target with time domain SNR being $-33.9$ dB. After coherent integration in fast time and slow time domain, the integrated SNRs of the three targets are $15.2$ dB, $15.2$ dB and $5.2$ dB, respectively. Here we set the false alarm rate as ${ P}_{\rm FA} = 0.1/(128\times 64)$, which is a little higher Compared to the false alarm rate ${ P}_{\rm FA} = 0.01/(128\times 64)$ setting in the previous simulation.

\begin{figure}
	\centering
        \subfigure[]{
	\label{MNOMP_result}
	\includegraphics[width = 60mm]{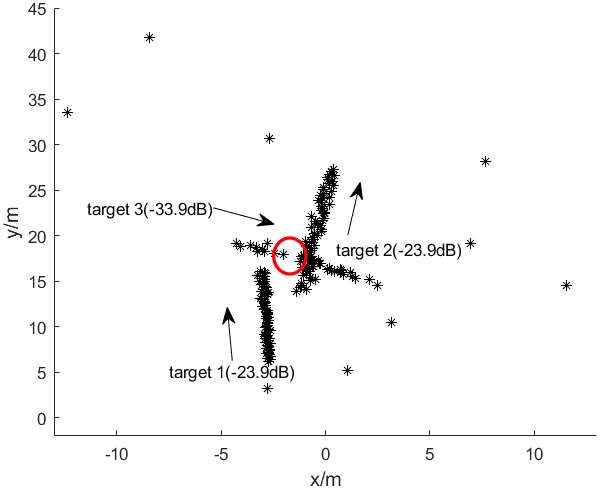}
    \label{MNOMP_meas}}
	% \subfigure[]{
	% \label{ifdetectweak}
	% \includegraphics[width = 60mm]{strong2_weak1_ifdetect.png}}
	\subfigure[]{
	\label{trackweak}
	\includegraphics[width = 60mm]{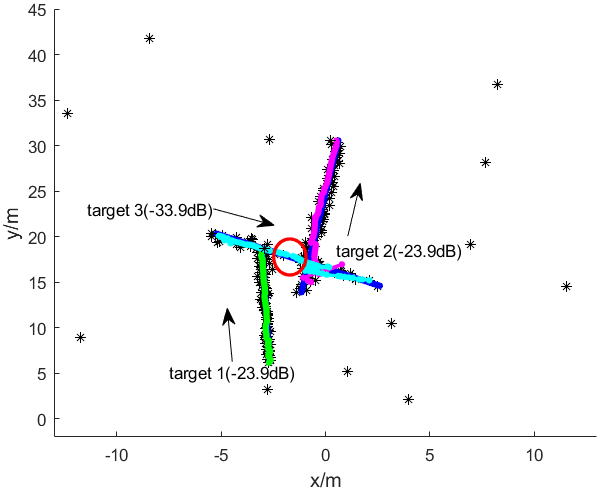}
    \label{MNOMP_spa_meas}}
 %        \subfigure[]{
	% \label{ifdetecttrack}
	% \includegraphics[width = 60mm]{strong2_weak1_track_detect.png}}
	\caption{The results of MNOMP and MNOMP-SPA-KF in the scene where two strong targets and one weak target coexist. (a) The detection and estimation of MNOMP,
    %(b) The detected mark result of MNOMP. 
    (b) The tracking result of MNOMP-SPA-KF. 
    %(d) The detected mark result of MNOMP-SPA-KF.
    }\label{strongweak} 
\end{figure}

As shown in Fig. \ref{MNOMP_meas}, the weak target is not detected by 2D-MNOMP at certain time, see the labeled red circle. However, in Fig. \ref{MNOMP_spa_meas}, the 2DMNOMP-SPA-KF tracking algorithm successfully maintains the trajectory, thereby enhancing the detection probability. Moreover, the false alarms can be suppressed by tracks.

\section{real data experiment}\label{realexp}
\begin{figure}
	\centering
	\includegraphics[width = 90mm]{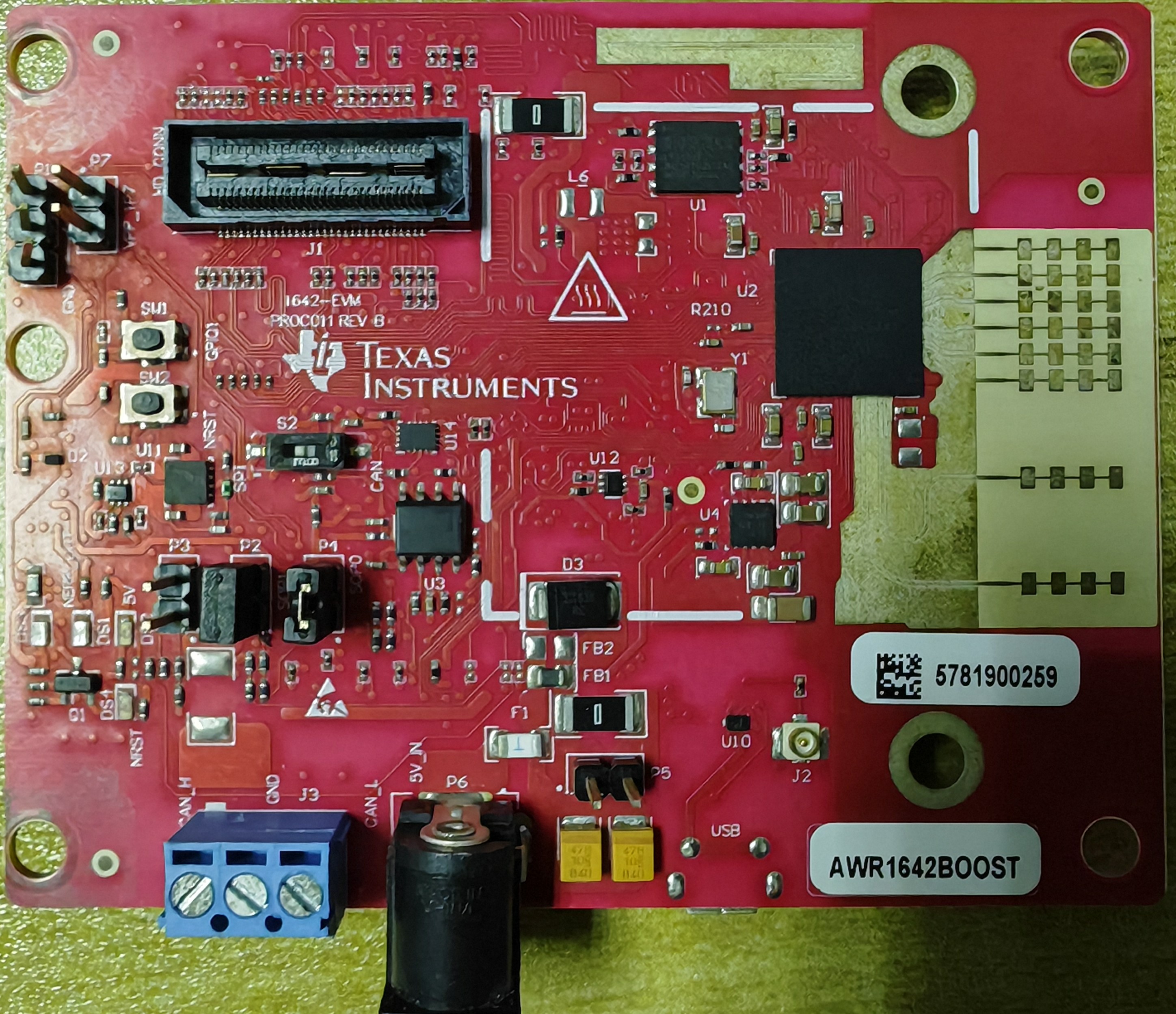}
	\caption{Physical image of AWR1642 radar.}\label{AWR1642}
\end{figure}
In this section, the performance of the proposed MNOMP-SPA-KF is demonstrated by conducting real experiments using AWR-1642 radar and a physical image is shown as Fig. \ref{AWR1642}. The AWR1642 radar is a FMCW MIMO radar consisting of two transmitters and four receivers. The radar parameters and waveform specifications are listed in Table \ref{tab:Radar_parameters}.

\newcommand{\tabincell}[2]{\begin{tabular}{@{}#1@{}}#2\end{tabular}}
\begin{table}[!t]
\centering
\scriptsize
\caption{Parameters Setting of the Experiment}
\label{tab:Radar_parameters}
\begin{tabular}{ll}
    \\[-2mm]
    \hline
    \hline\\[-2mm]
    { \small Parameters}&\qquad {\small Value}\\
    \hline
    \vspace{1mm}\\[-3mm]
    Carrier frequency $f_c$   &   77 GHz\\
    \vspace{1mm}
    Frequency modulation slope $\mu$          &  $8.012$ MHz/$\mu$s\\
    \vspace{1mm}
    sweep time $T_p$         &  $56 {\rm \mu s}$\\
    \vspace{1mm}
    Pulse repeat interval $T_{\text{r}}$         &  $3 {\rm \mu s}$\\
    \vspace{1mm}
    Bandwidth $B$          &  448.672 MHz\\
    \vspace{1mm}
    Sampling frequency $f_s = 1 / T_s$         &  5 MHz\\
    \vspace{1mm}
    Number of fast time samples $N$          &  128 \\
    \vspace{1mm}
    Number of pulses in one CPI $M$          & 64\\
    \vspace{1mm}
    Number of Receivers $L$      &   4\\
    \vspace{1mm}
    observation time interval $T$    &   0.1s\\
	\vspace{1mm}
	maximum range $r_{\rm max}$ & 93.6m \\
	\vspace{1mm}
	maximum measurable radial velocity $v_{\rm max}$ & 16.5 m/s \\
	\vspace{1mm}
	maximum field of view $\theta_{\rm max}$ & $\pi / 3$ \\
    [1mm]
    \hline
    \hline
\end{tabular}
\end{table}

The default parameters of the MNOMP-SPA-KF are set as follows: false alarm rate ${\rm P}_{\rm FA} = 0.01/(128 \times64)$, the number of training cells is $60$, the threshold multiplier can be calculated as $5.4783$, i.e., $7.3865$ dB, the gate probability is ${P}_{\rm G} = 0.95$, the number of SPA iterations is $10$. In practical situations, 3D-NOMP generates many false alarms and is very time consuming. Consequently, we employ 2D-MNOMP to process the baseband data.

Because the targets such as pedestrians, bicycles, and cars typically generate multiple measurements, clustering algorithm is employed to preprocess the measurements which are the output of the 2D-MNOMP algorithm. We compute the mean position and velocity of each measurement within each cluster to obtain a new measurement set.
% The clustering algorithm we use is described as follows:
% We subtract the position and velocity three-dimensional vectors of the first measurement from the remaining measurements in the measurement set and compute the square of each component.
% We then sum the weighted squares to obtain the differences between them.
% If the difference between some measurements and the first measurement is less than a threshold, we consider the first measurement and these measurements as one cluster and remove them from the measurement set.
% We then repeat the same process for the first measurement in the remaining measurement set until there are no measurements left.
% Finally, 

% the upper bound of targets number $K_{\rm max} = 16$;

\subsection{Experiment I}
Fig. \ref{Exp1scene} shows the setup of field experiment I.
In experiment I, people $1$ walked from $(-1.50m, 1.24m)$ to $(2.50m, 4.74m)$, while people $2$ walked from $(-2.00m, 5.24m)$ to $(2.00m, 1.24m)$.
The total observation time is $6.4$s, which corresponds to $64$ frames.

Fig. \ref{exp1MNOMP_PDA} $ \sim $ Fig. \ref{exp1MNOMP_SPA_subKF} illustrate the target trajectories produced by the benchmark algorithms and the proposed MNOMP-SPA-KF. 
% From Fig. \ref{exp1FFT_PDA_KF} $ \sim $ Fig. \ref{exp1FFT_SPA_subKF}, we can see that FFT based mathod can not track the two people successfully due to the influence of clutters. 
Fig. \ref{exp1MNOMP_PDA} shows that MNOMP-PDA-KF can track the two people, but also generate many false tracks. Fig. \ref{exp1MNOMP_PDA_subKF} shows that MNOMP-PDA-subKF fails to track the two people. Fig. \ref{exp1MNOMP_JPDA_KF} shows that MNOMP-JPDA-subKF can track the two people without generating many false tracks, but the error between the trajectories and the actual paths is high when the trajectories intersect. Fig. \ref{exp1MNOMP_SPA_subKF} shows that MNOMP-SPA-subKF can track people 1, but fails to track people 2 all the time.

The tracking results of MNOMP-SPA-KF (2D-Gate) with MNOMP-SPA-KF (3D-Gate) are shown in Fig. \ref{exp1Res_2DGate} and Fig. \ref{exp1Res_3DGate}, respectively. The MNOMP-SPA-KF algorithm successfully tracks two people, with the estimated trajectories closely aligning with the actual paths and effectively suppressing false alarms.
Comparing Fig. \ref{exp1Res_2DGate} with Fig. \ref{exp1Res_3DGate}, we observe that MNOMP-SPA-KF (2D-Gate) is more susceptible to clutter and interference from another target during the tracking process, resulting in deviations and interruptions in the trajectories at the intersection of the two people.
However, MNOMP-SPA-KF (3D-Gate) eliminates interference from noise and the other target, ensuring stable trajectories even at intersections.

\begin{figure}
\centering
\includegraphics[width = 120mm]{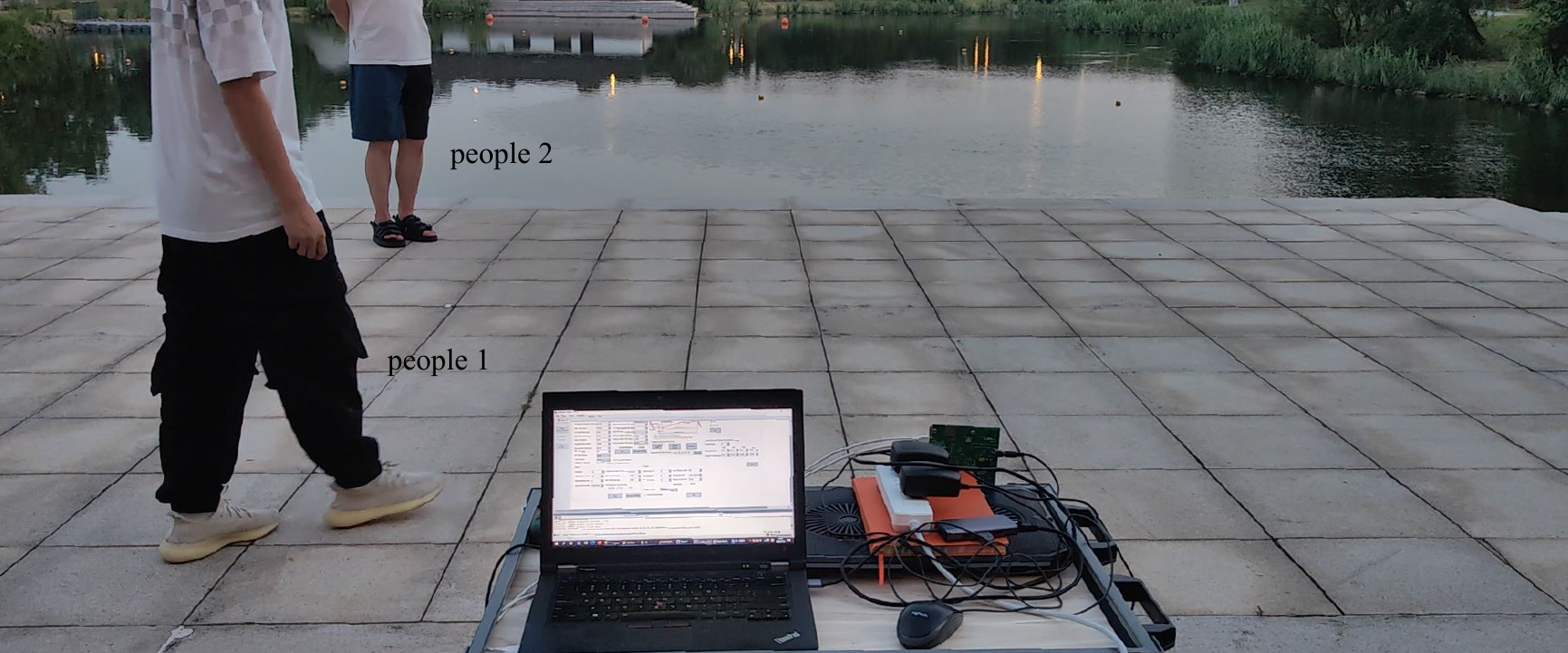}
\caption{The scene and original state of experiment I.}\label{Exp1scene}
\end{figure}
\begin{figure}
\centering
\subfigure[MNOMP-PDA-KF.]{
\label{exp1MNOMP_PDA}
\includegraphics[width = 50mm]{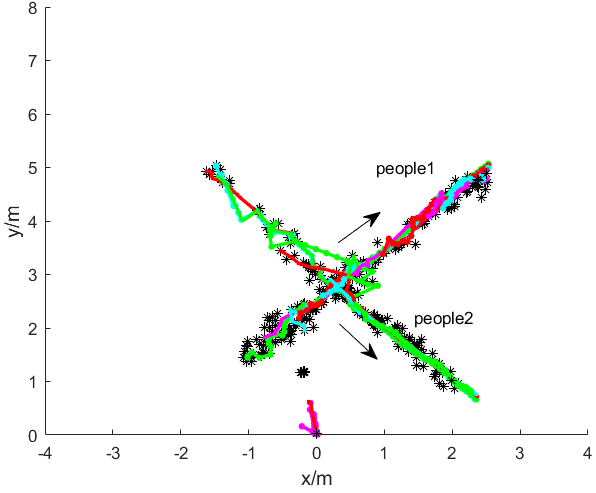}}
\subfigure[MNOMP-PDA-subKF.]{
\label{exp1MNOMP_PDA_subKF}
\includegraphics[width = 50mm]{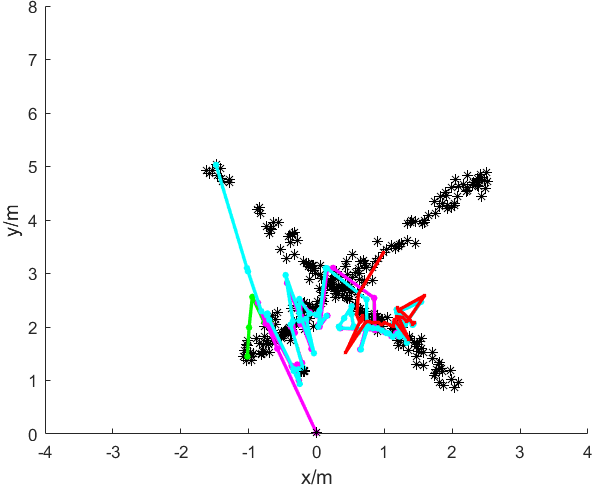}}
\subfigure[MNOMP-JPDA-subKF.]{
\label{exp1MNOMP_JPDA_KF}
\includegraphics[width = 50mm]{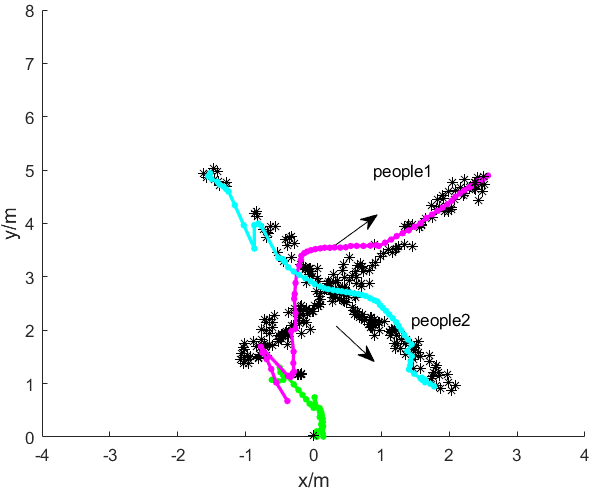}}
\subfigure[MNOMP-SPA-subKF.]{
\label{exp1MNOMP_SPA_subKF}
\includegraphics[width = 50mm]{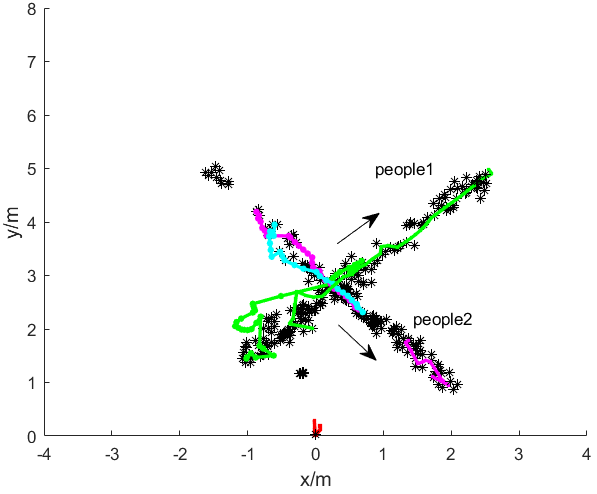}}
\subfigure[\scriptsize MNOMP-SPA-KF with 2D-Gate.]{
\label{exp1Res_2DGate}
\includegraphics[width = 50mm]{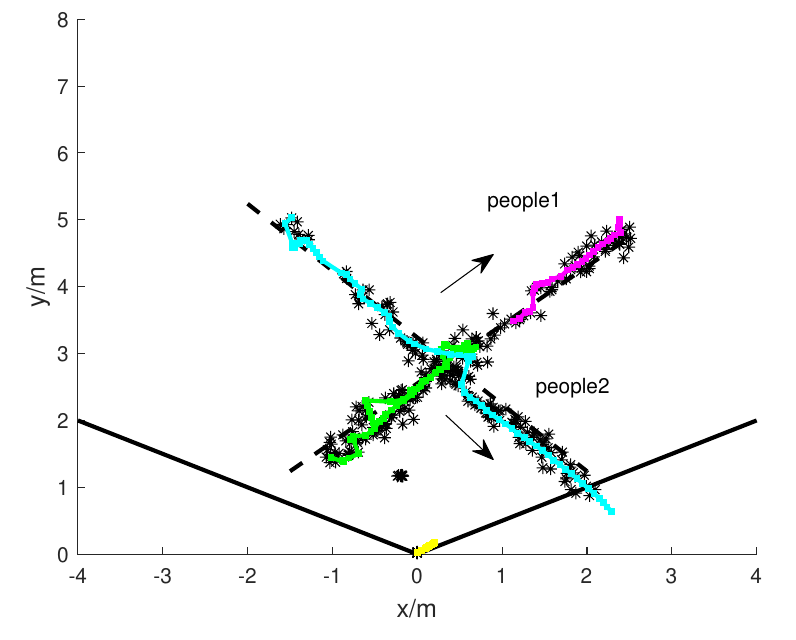}}
\subfigure[\scriptsize MNOMP-SPA-KF with 3D-Gate.]{
\label{exp1Res_3DGate}
\includegraphics[width = 50mm]{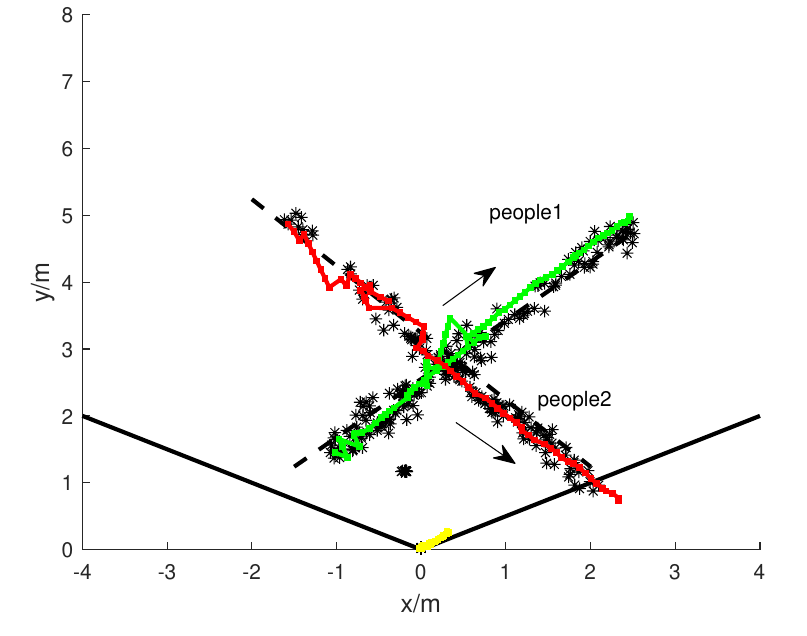}}
\caption{The tracking results of different algorithms. The arrow indicates the direction of the target movement.} \label{exp1Res}
\end{figure}

% The true trajectory, measurements and trakcing result are shown in the Fig. \ref{exp1Res}.
% It can be found from the Fig. \ref{exp1Res} that the trajectories estimated by SFNC-SPA algorithm match the ground truth of the people moving trajectories roughly and retain resolution property even if the two trajectories intersect with each other.
% However, the people couldn't be treated as point target simply and the SFNC\_T algorithm gives far more measurements than the number of targets, there exists many false trajectories in the tracking results.

% Compare the result from ``hard'' judgement and ``soft'' judgement in Fig.\ref{exp1Res_HardClu} and Fig.\ref{exp1Res_SoftClu}, respectively.
% We find that the trajectory obtained by SFNC-SPA (soft) is less likely to break down under the condition that the target doesn't produce any measurements.
% But the trajectory caused by clutters and flase alarms is more difficult to disappear under the ``soft'' judgement.

\subsection{Experiment II}

Fig. \ref{Exp2scene} shows the experimental setup for Experiment II, which includes two people, one cyclist, and a stationary car.
The two people are moving towards the radar, while the cyclist initially moves away from the radar before returning towards it.
Fig. \ref{exp2MNOMP_SPA} $ \sim $  Fig. \ref{exp2MNOMP_JPDA_KF} illustrate the target trajectories produced by the benchmark algorithms and the proposed MNOMP-SPA-KF. From Fig. \ref{exp2MNOMP_SPA}, we can observe that MNOMP-SPA-KF algorithm roughly outlines the trajectories of the two people and the cyclist. Fig. \ref{exp2MNOMP_SPA_subKF} $ \sim $  Fig. \ref{exp2MNOMP_JPDA_KF} show that MNOMP-SPA-subKF and MNOMP-PDA-KF track the two people and the cyclist, but generate more false tracks compared to MNOMP-SPA-KF. MNOMP-PDA-subKF and MNOMP-JPDA-subKF fails to track people 2 and also generate many false tracks.
% From Fig. \ref{exp2FFT_SPA} $ \sim $ Fig. \ref{exp2FFT_JPDA_KF}, we can see that FFT based mathod can not track the two people successfully due to the influence of clutters.}

The entire scenario spans 10 seconds, comprising a total of 100 frames of data.
After employing the MNOMP-SPA-KF algorithm with both 2D-Gate and 3D-Gate for sensing and tracking, we have extracted and displayed the tracking results for the 25th, 50th, 75th and 100th frames in Fig. \ref{exp2Res}, which correspond to the estimated trajectories at 2.5s, 5s, 7.5s, and 10s, respectively.
To make the results intuitive and concise, we omit trajectories of targets that have disappeared at the time of plotting.

\begin{figure}
\centering
\includegraphics[width = 100mm]{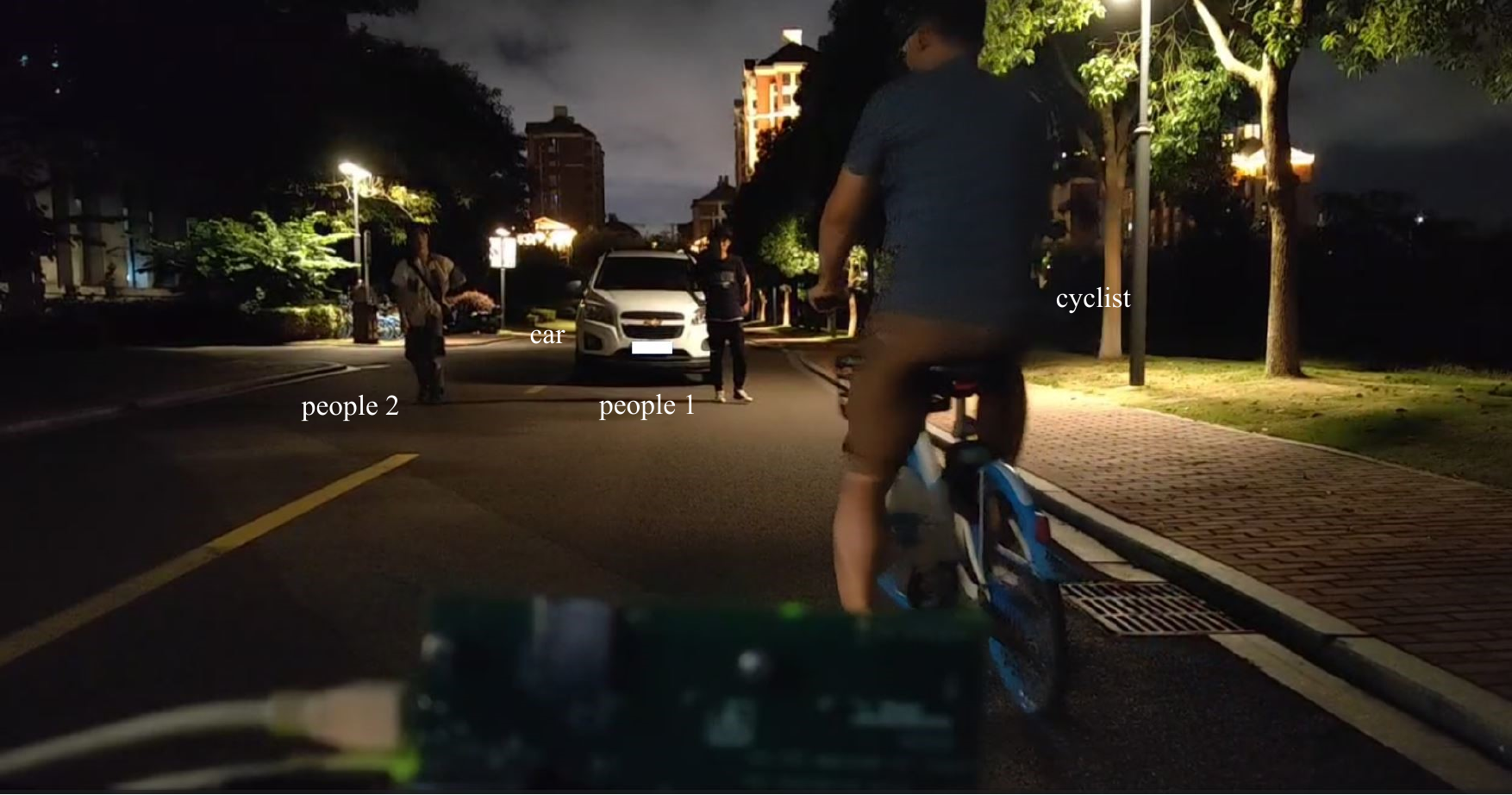}
\caption{The scene and original state of experiment II.}\label{Exp2scene}
\end{figure}

\begin{figure}
\centering
\subfigure[\scriptsize MNOMP-SPA-KF.]{
\label{exp2MNOMP_SPA}
\includegraphics[width = 48mm]{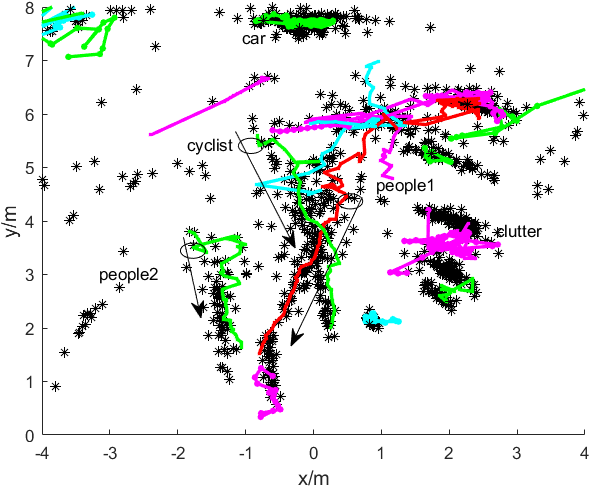}}
\subfigure[\scriptsize MNOMP-SPA-subKF.]{
\label{exp2MNOMP_SPA_subKF}
\includegraphics[width = 48mm]{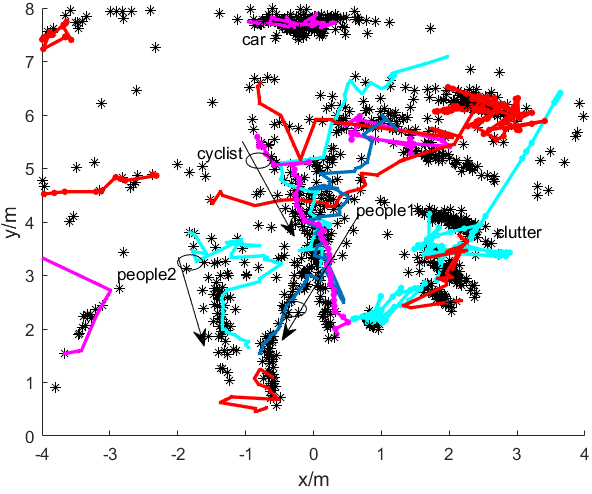}}
\subfigure[\scriptsize MNOMP-PDA-KF.]{
\label{exp2MNOMP_PDA}
\includegraphics[width = 48mm]{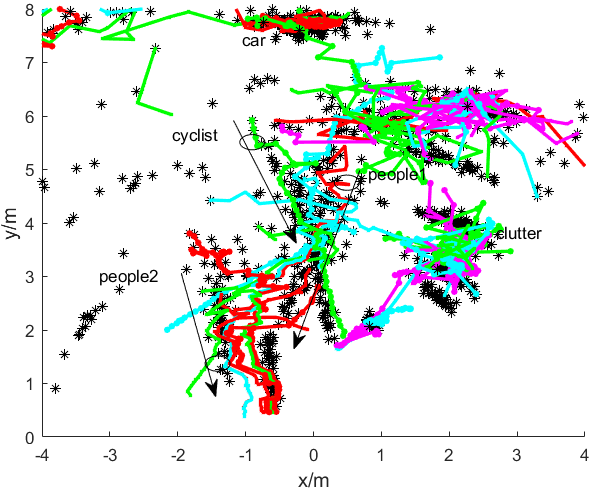}}
\subfigure[\scriptsize MNOMP-PDA-subKF.]{
\label{exp2MNOMP_PDA_subKF}
\includegraphics[width = 48mm]{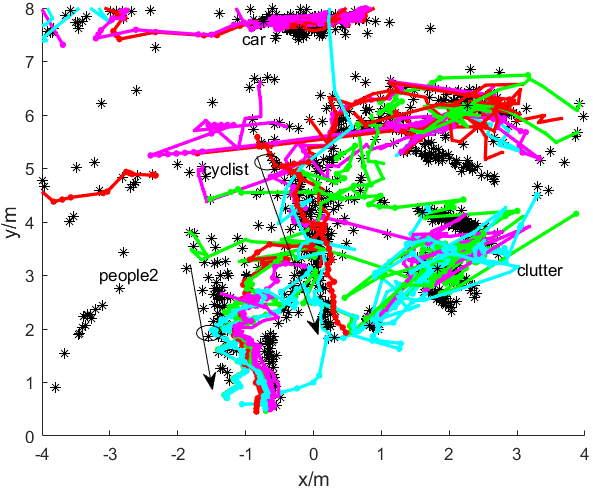}}
\subfigure[\scriptsize MNOMP-JPDA-subKF.]{
\label{exp2MNOMP_JPDA_KF}
\includegraphics[width = 48mm]{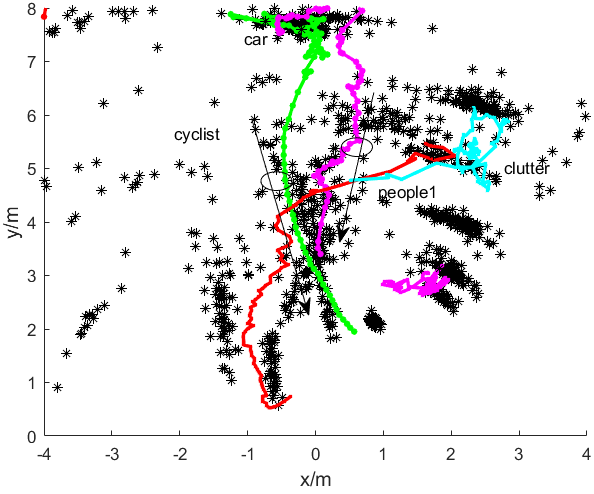}}
\caption{The tracking results of experiment II by different algorithms. The arrow indicates the direction of the target movement.}\label{exp2Res}
\end{figure}

\begin{figure}
\centering
\subfigure[\tiny MNOMP-SPA-KF (2D-Gate) at 2.5s.]{
\label{expII2D_25}
\includegraphics[width = 38mm]{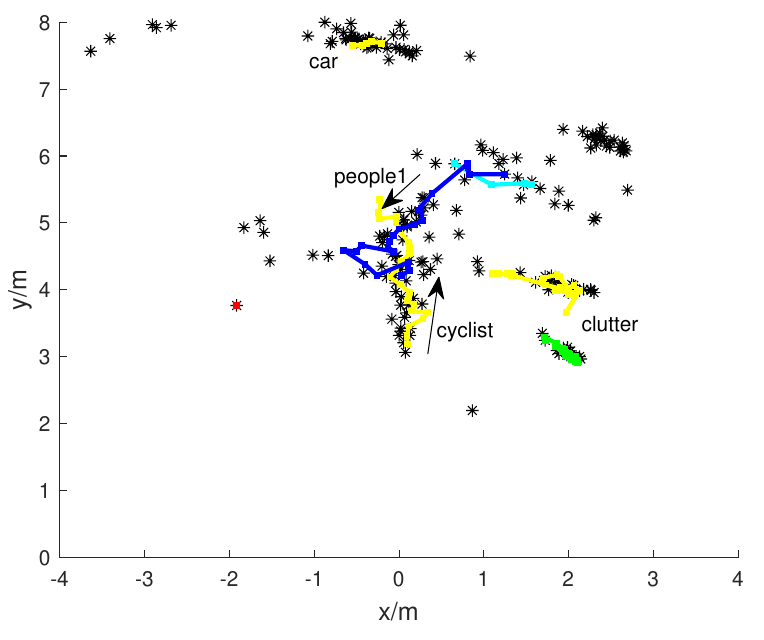}}
\subfigure[\tiny MNOMP-SPA-KF (2D-Gate) at 5s.]{
\label{expII2D_50}
\includegraphics[width = 38mm]{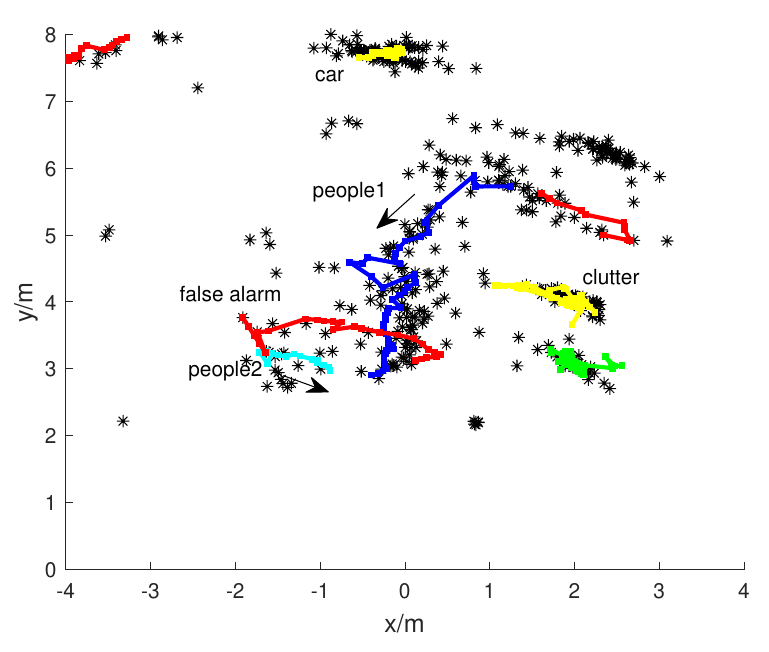}}
\subfigure[\tiny MNOMP-SPA-KF (2D-Gate) at 7.5s.]{
\label{expII2D_75}
\includegraphics[width = 38mm]{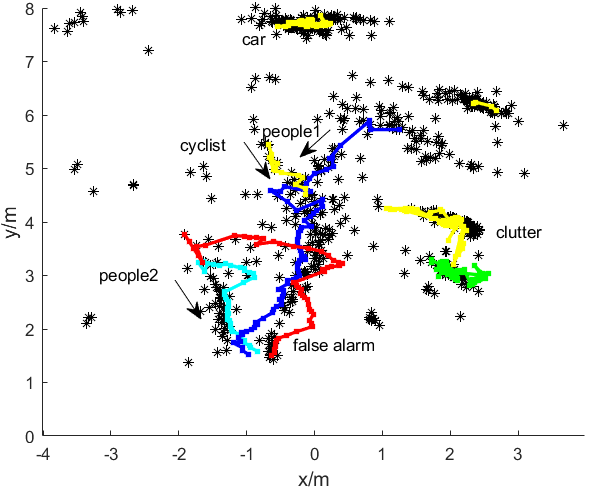}}
\subfigure[\tiny MNOMP-SPA-KF (2D-Gate) at  10s.]{
\label{expII2D_100}
\includegraphics[width = 38mm]{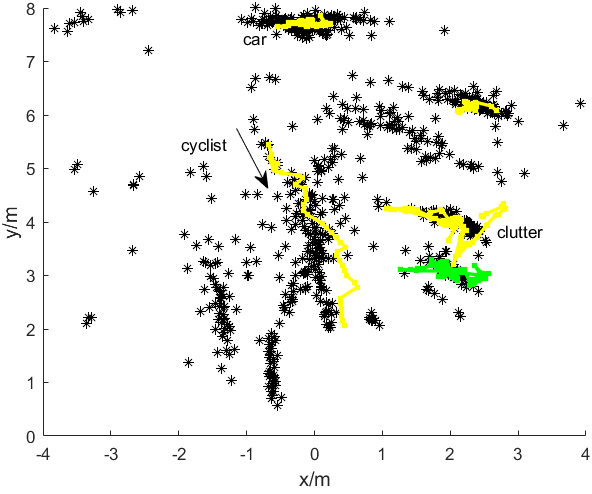}}
\subfigure[\tiny MNOMP-SPA-KF (3D-Gate) at 2.5s.]{
\label{expII3D_25}
\includegraphics[width = 38mm]{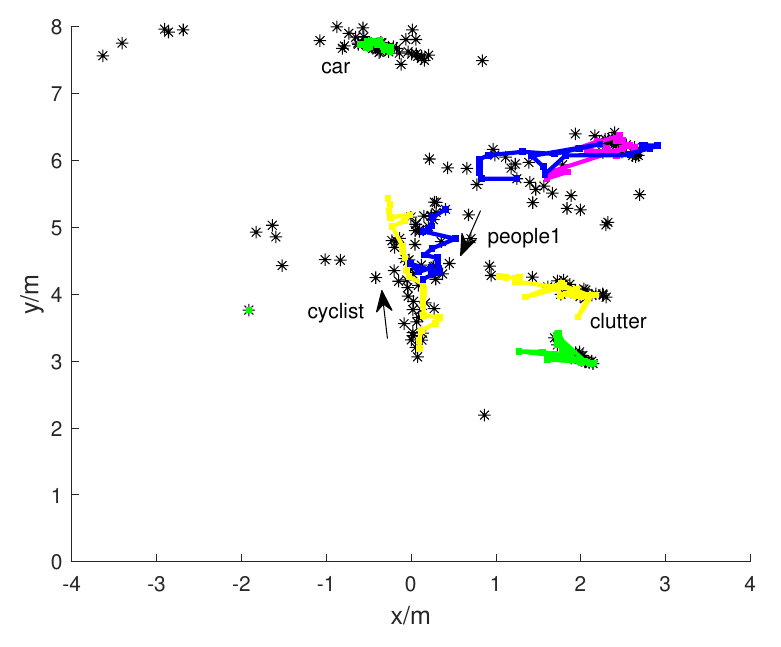}}
\subfigure[\tiny MNOMP-SPA-KF (3D-Gate) at 5s.]{
\label{expII3D_50}
\includegraphics[width = 38mm]{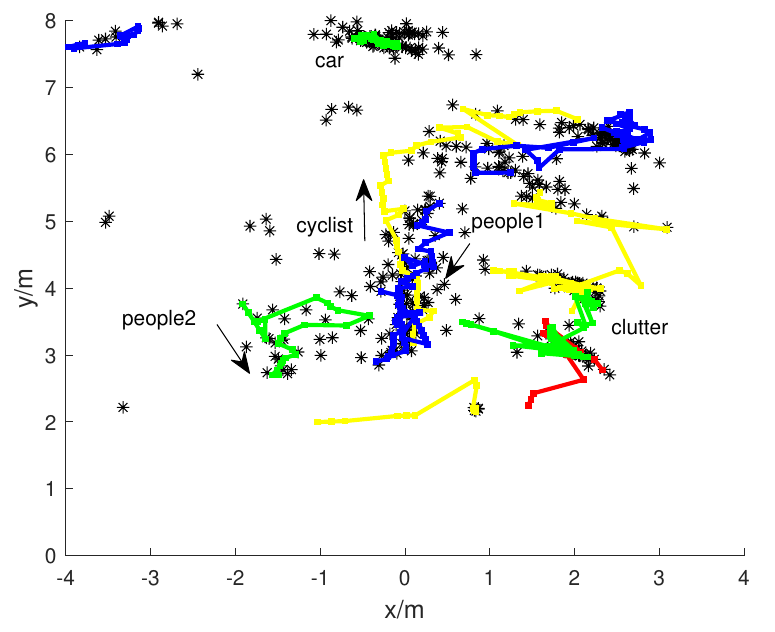}}
\subfigure[\tiny MNOMP-SPA-KF (3D-Gate) at 7.5s.]{
\label{expII3D_75}
\includegraphics[width = 38mm]{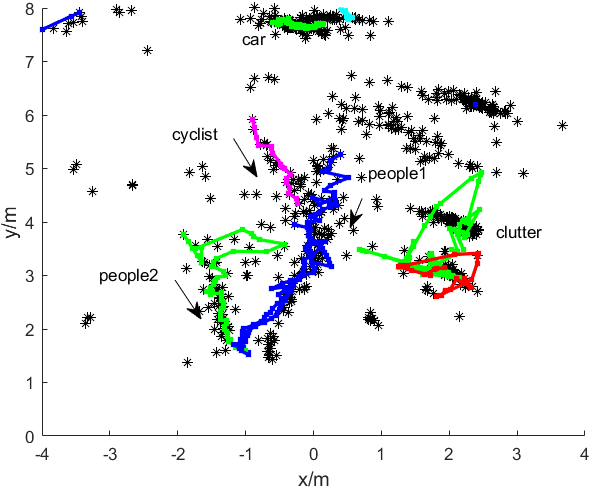}}
\subfigure[\tiny MNOMP-SPA-KF (3D-Gate) at 10s.]{
\label{expII3D_100}
\includegraphics[width = 38mm]{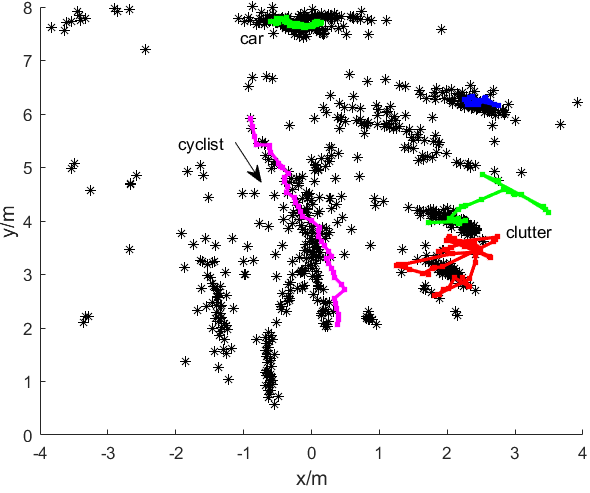}}
\caption{The tracking results of experiment II using 2D-Gate and 3D-Gate in $2.5$s, $5$s, $7.5$s and $10$s, respectively. The arrow indicates the direction of the target movement.}\label{exp2Res}
\end{figure}

The detailed tracking results of MNOMP-SPA-KF algorithm with both 2D-Gate and 3D-Gate are shown in Fig. \ref{exp2Res}. In the first $2.5$s, the cyclist moves away from the radar while people 1 approaches it, depicted by two separate trajectories in both Fig. \ref{expII2D_25} and Fig. \ref{expII3D_25}.
In 2.5s-5s, both people continue to move towards the radar while the cyclist begins to return, resulting in the disappearance of trajectory representing the cyclist in Fig. \ref{expII2D_50} and Fig. \ref{expII3D_50}.
Fig. \ref{expII3D_50} further illustrates the approximate trajectory of the cyclist during the return phase using MNOMP-SPA-KF (3D-Gate).
In 5s-7.5s, the motion states of both people remain unchanged while the cyclist returns and approaches the radar.
In Fig. \ref{expII2D_75} and Fig. \ref{expII3D_75}, trajectories of all three targets are generated simultaneously, but MNOMP-SPA-KF (2D-Gate) exhibits mismatch and false alarms in estimating the trajectory of people 2, while MNOMP-SPA-KF (3D-Gate) does not have these issues.
In 7.5s-10s, both people move out of the radar's observation range, leaving only the cyclist's trajectory visible in Fig. \ref{expII2D_100} and Fig. \ref{expII3D_100}.

Comparing MNOMP-SPA-KF (2D-Gate) with MNOMP-SPA-KF (3D-Gate), it is evident that MNOMP-SPA-KF (2D-Gate) generates mismatches between targets and measurements during tracking, leading to false alarms.
Conversely, the MNOMP-SPA-KF (3D-Gate) can mitigate interference from clutter and other targets during tracking, leading to less mismatches and false alarms. Therefore, MNOMP-SPA-KF (3D-Gate) performs better.

% Fig. \ref{Exp2scene} shows the original state of experiment II, which contains two people, a cyclist and a stationary car.
% The two people walked towards the radar, the cyclist rode away from radar and then turn back to ride towards the radar.
% The total observation time is $10$s, which corresponds to $100$ frames.
% We display the tracking result by SFNC-SPA (soft) at the moments of $2.5$s, $5$s, $7.5$s and $10$s in Fig. \ref{exp2Res}, to make the result intuitive and concise, we ignore the trajectories of vanished targets at the moment of display.

% At the momnet of $2.5$s, we can see the trajectories of people $1$ and the cyclist in Fig.\ref{exp2Res_25}, but at this time, people $2$ is occluded and not detected by the radar.
% Therefore, his trajectory is not detected by SFNC-SPA algorithm.
% At the moment of $5$s, the cyclist who was truning around is not detected because of the CV model applied by SFNC-SPA algorithm.
% At this moment, we can only see the trajectories of people 1 and people 2.
% At the moment of $7.5$s, we can see all the trajectories of people 1, people 2 and the cyclist in Fig.\ref{exp2Res_75}.
% And at the moment of $10$s, since people 2 left the observation region of radar and people 1 changed his direction, we can only track the cyclist stably.
% During the observation time, the stationary car can be detected all the time.
% Because the car will generate several measurements at the same moment, which may locate differently at different moments, the trajectory of car is estimated by the SFNC-SPA.

\subsection{Experiment III}
Fig. \ref{Exp3scene} shows the experimental setup for experiment III, which contains two people, and a moving car. The car approaches the radar while people 1 walks away from the radar, and people 2 walks towards the radar. The total observation time is also $10$s.

\begin{figure}
\centering
\includegraphics[width = 100mm]{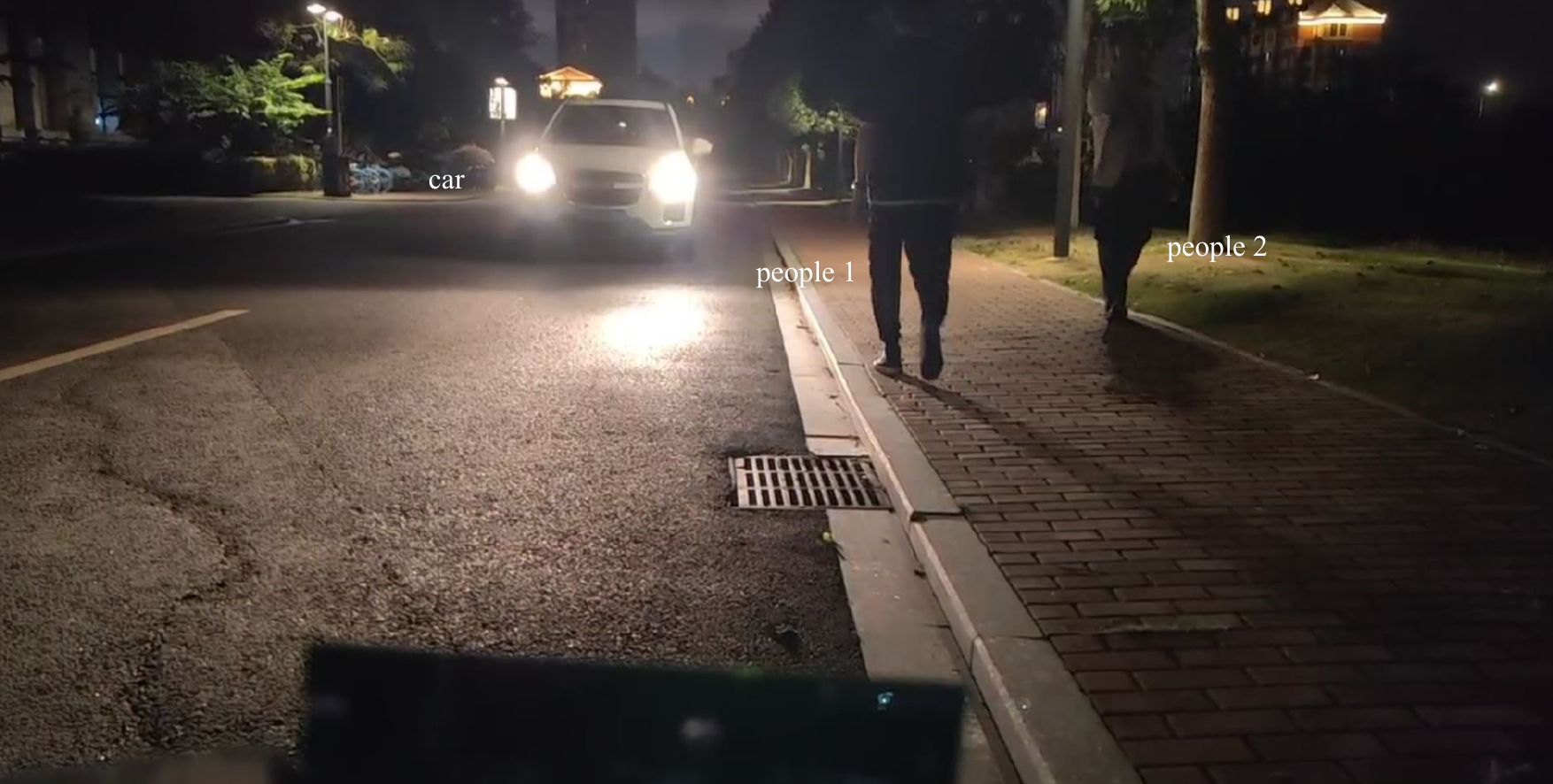}
\caption{The scene and original state of experiment III.}\label{Exp3scene}
\end{figure}
% \subsection{Experiment IV}

\begin{figure}
\centering
\subfigure[\scriptsize MNOMP-SPA-KF.]{
\label{exp3MNOMP_SPA}
\includegraphics[width = 48mm]{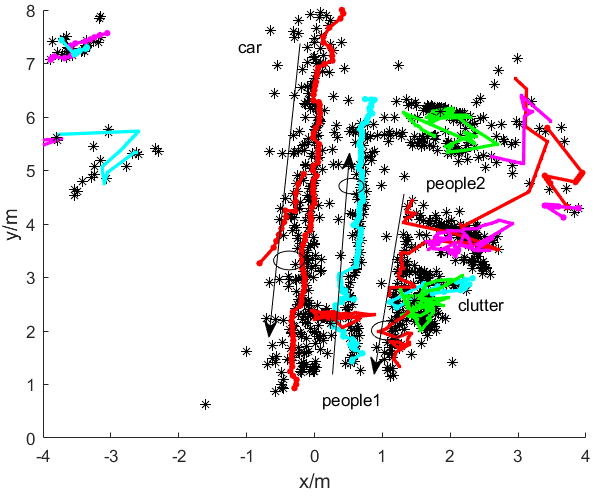}}
\subfigure[\scriptsize MNOMP-SPA-subKF.]{
\label{exp3MNOMP_SPA_subKF}
\includegraphics[width = 48mm]{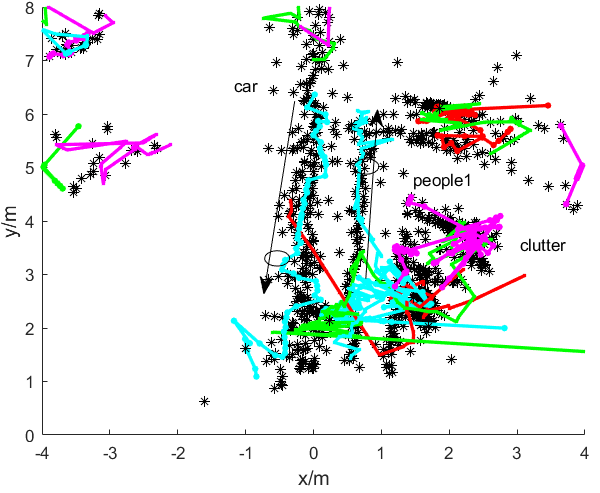}}
\subfigure[\scriptsize MNOMP-PDA-KF.]{
\label{exp3MNOMP_PDA}
\includegraphics[width = 48mm]{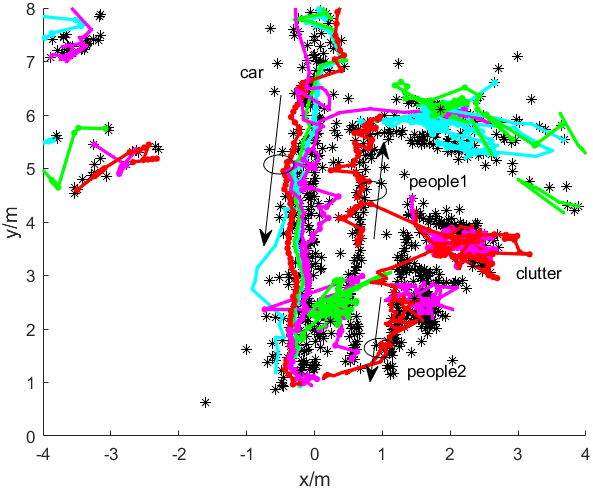}}
\subfigure[\scriptsize MNOMP-PDA-subKF.]{
\label{exp3MNOMP_PDA_subKF}
\includegraphics[width = 48mm]{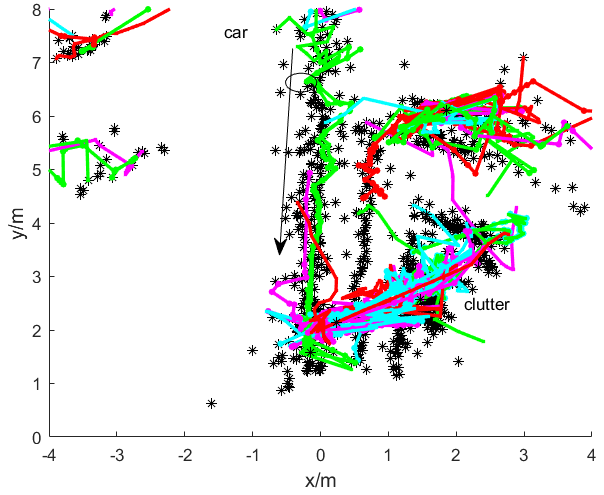}}
\subfigure[\scriptsize MNOMP-JPDA-subKF.]{
\label{exp3MNOMP_JPDA_KF}
\includegraphics[width = 48mm]{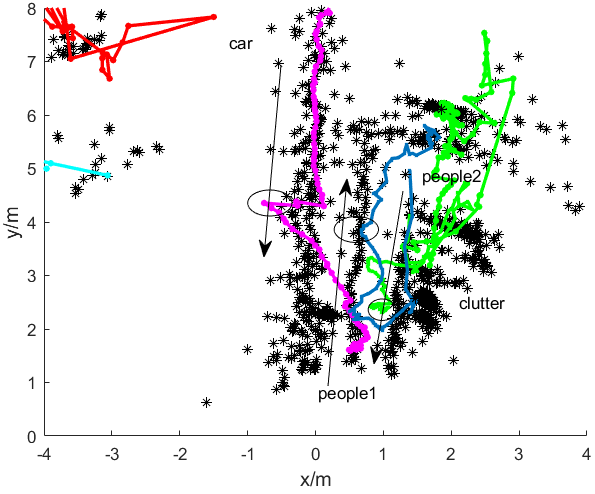}}
\caption{The tracking results of experiment III by different algorithms. The arrow indicates the direction of the target movement.}\label{exp3Res}
\end{figure}

\begin{figure}
\centering
\subfigure[\tiny MNOMP-SPA-KF (2D-Gate) at 2.5s.]{
\label{expIII2D_25}
\includegraphics[width = 38mm]{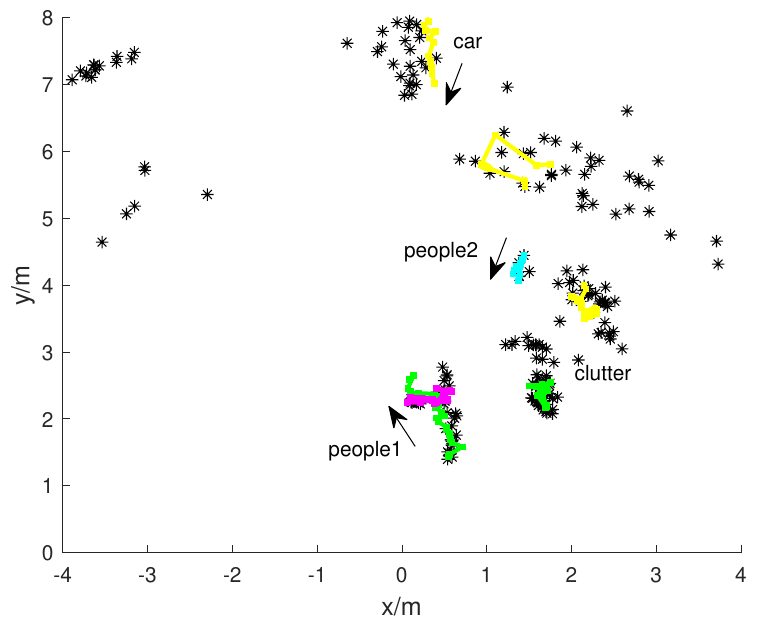}}
\subfigure[\tiny MNOMP-SPA-KF (2D-Gate) at 5s.]{
\label{expIII2D_50}
\includegraphics[width = 38mm]{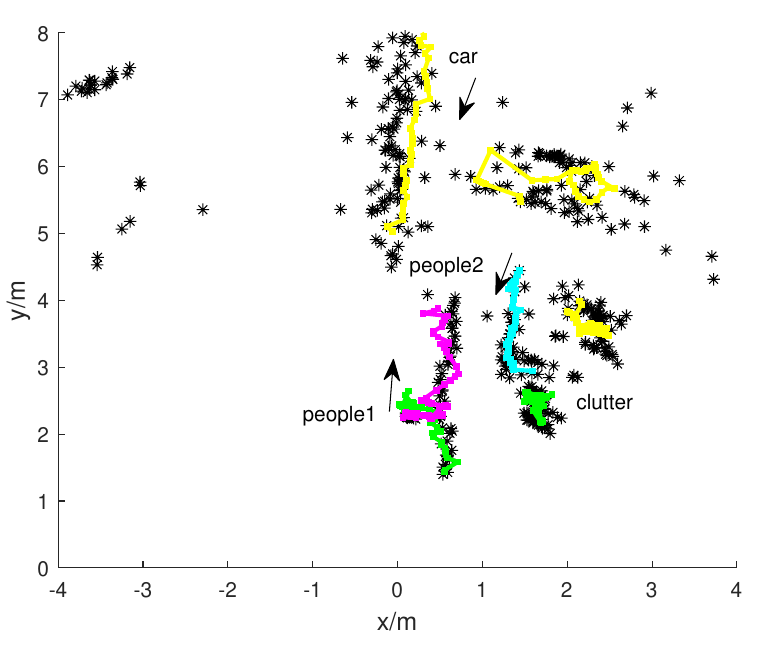}}
\subfigure[\tiny MNOMP-SPA-KF (2D-Gate) at 7.5s.]{
\label{expIII2D_75}
\includegraphics[width = 38mm]{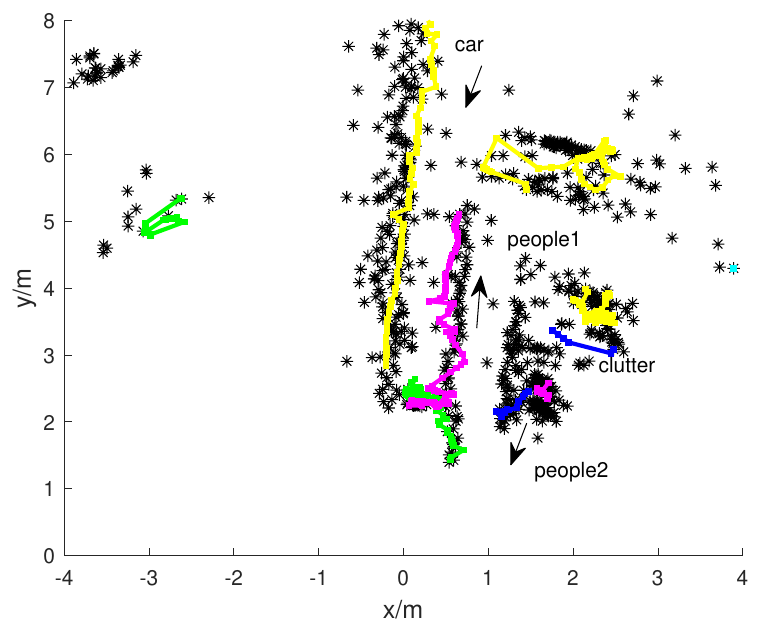}}
\subfigure[\tiny MNOMP-SPA-KF (2D-Gate) at 10s.]{
\label{expIII2D_100}
\includegraphics[width = 38mm]{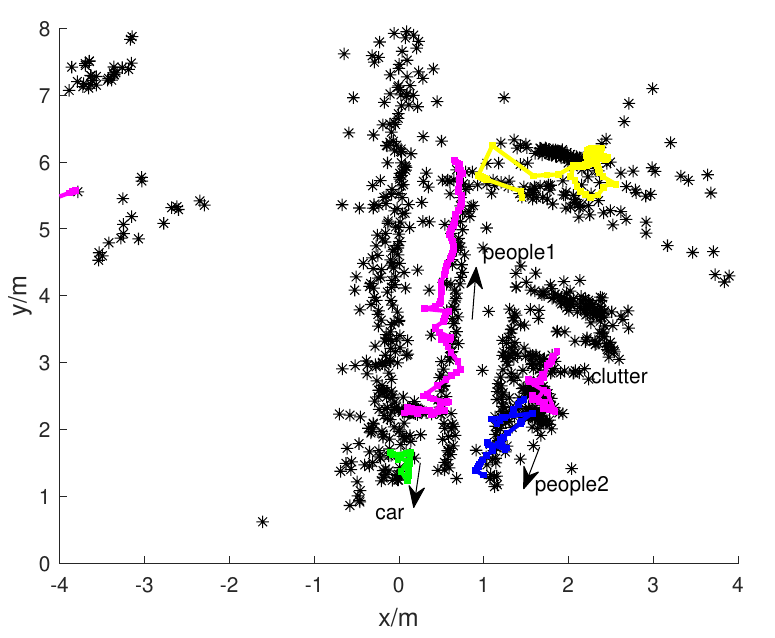}}
\subfigure[\tiny MNOMP-SPA-KF (3D-Gate) at 2.5s.]{
\label{expIII3D_25}
\includegraphics[width = 38mm]{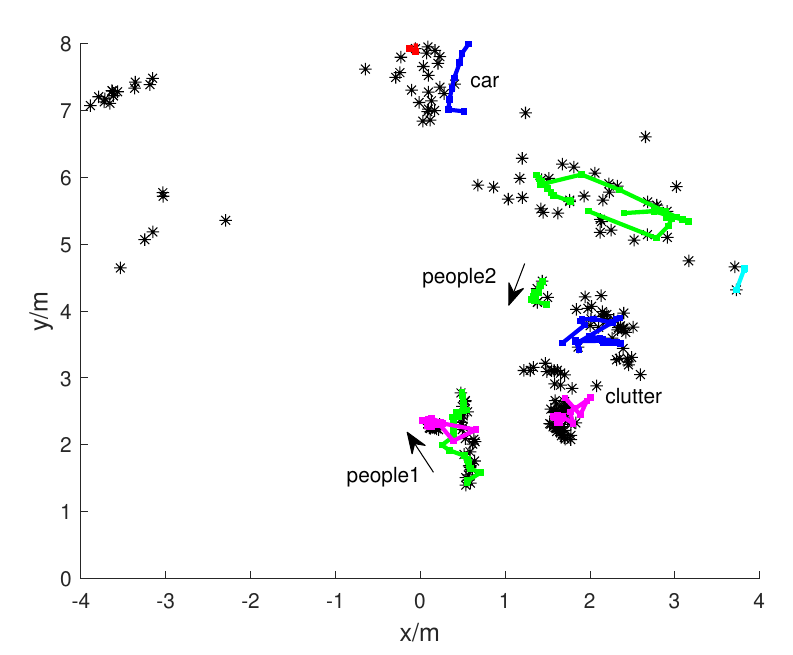}}
\subfigure[\tiny MNOMP-SPA-KF (3D-Gate) at 5s.]{
\label{expIII3D_50}
\includegraphics[width = 38mm]{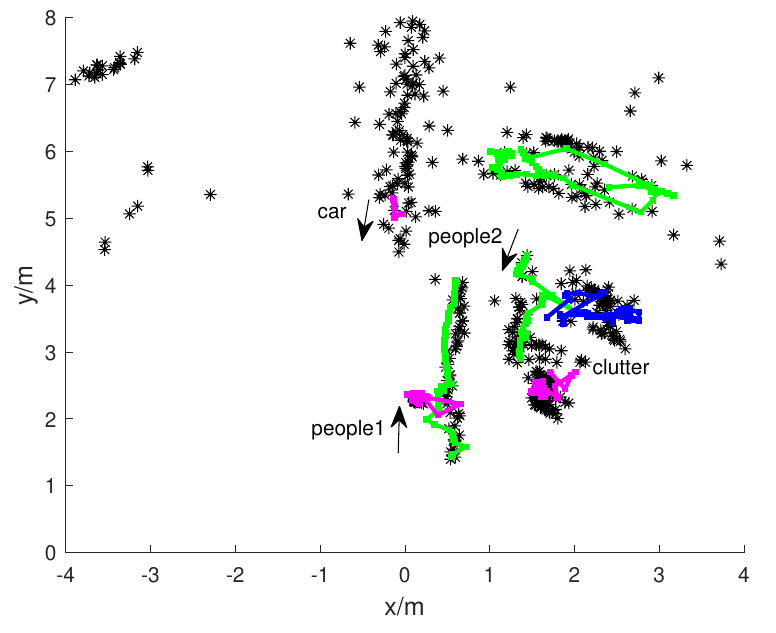}}
\subfigure[\tiny MNOMP-SPA-KF (3D-Gate) at 7.5s.]{
\label{expIII3D_75}
\includegraphics[width = 38mm]{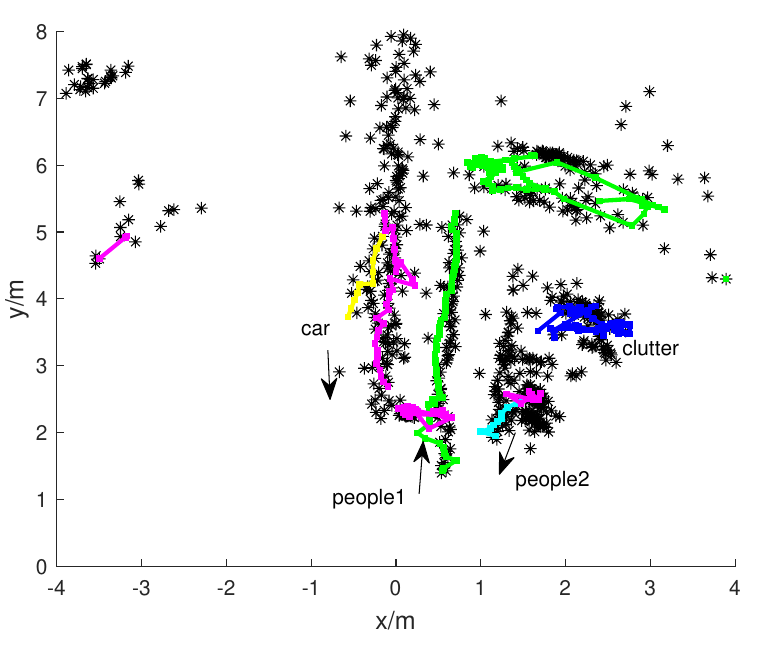}}
\subfigure[\tiny MNOMP-SPA-KF (3D-Gate) at 10s.]{
\label{expIII3D_100}
\includegraphics[width = 38mm]{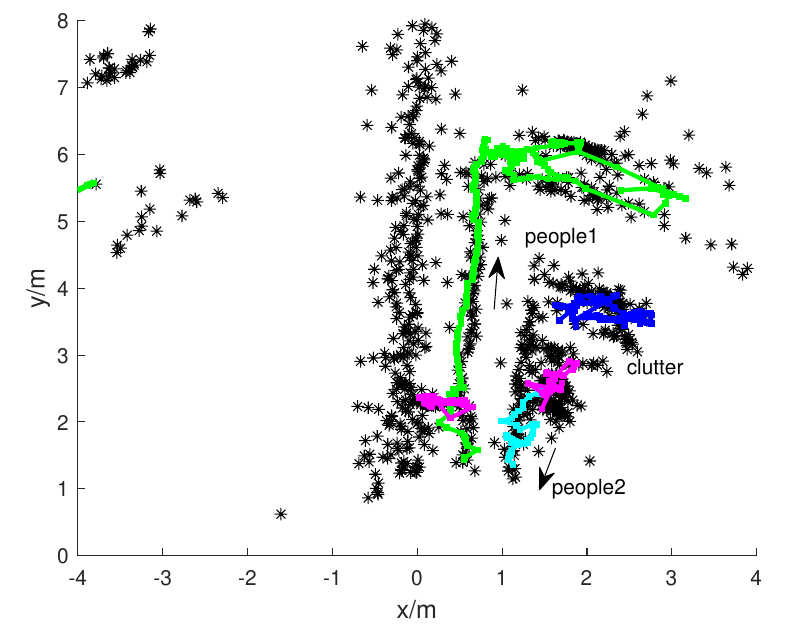}}
\caption{The tracking result of experiment III using 2D-Gate and 3D-Gate in $2.5$s, $5$s, $7.5$s and $10$s, respectively. The arrow indicates the direction of the target movement.}\label{exp3Res}
\end{figure}
% In Fig. \ref{exp3Res}, the trajectories of the car and two people provided by MNOMP-SPA-KF are displayed 
Fig. \ref{exp3MNOMP_SPA}-Fig. \ref{exp3MNOMP_JPDA_KF} illustrate the target trajectories produced by the the benchmark algorithms and the proposed MNOMP-SPA-KF.
From Fig. \ref{exp3MNOMP_SPA}, we observe that MNOMP-SPA-KF algorithm roughly outlines the trajectories of the two people and the car. Fig. \ref{exp3MNOMP_SPA_subKF} $ \sim $  Fig. \ref{exp3MNOMP_PDA_subKF} show that MNOMP-SPA-subKF, MNOMP-PDA-KF and MNOMP-PDA-subKF cause many false tracks and MNOMP-PDA-subKF fails to track the two people. Fig. \ref{exp3MNOMP_JPDA_KF} shows that MNOMP-JPDA-subKF has large tracking errors of the car and people 1 in the end.
% From Fig. \ref{exp2FFT_SPA} $ \sim $ Fig. \ref{exp2FFT_JPDA_KF}, we can see that FFT based mathod can not track the two people successfully due to the influence of clutters.}

We display the tracking result by MNOMP-SPA-KF at the moments of $2.5$s, $5$s, $7.5$s and $10$s, which are shown in Fig. \ref{expIII2D_25} $-$ Fig. \ref{expIII3D_100}.
During the tracking of people 1, MNOMP-SPA-KF (2D-Gate) results in disrupt of the trajectory, and generate a new trajectory at about $2$s 
shown in Fig. \ref{expIII2D_25},  which also can be seen in Fig. \ref{expIII2D_50}, whereas MNOMP-SPA-KF (3D-Gate) provides a complete trajectory.
In 7.5s-10s, the tracking error of people 2 produced by MNOMP-SPA-KF (2D-Gate) is slightly larger than that produced by MNOMP-SPA-KF (3D-Gate), as shown in Fig. \ref{expIII2D_75} and Fig. \ref{expIII2D_100}. 
It can be concluded that MNOMP-SPA-KF (3D-Gate) performs best.

To show the advantage of the proposed algorithm's time-consuming, we implement the original 2D-MNOMP consisting of the forward and backward steps, multiple 2D-FFTs. It is shown that the original 2D-MNOMP and the proposed 2D-MNOMP takes about 1.69s and 0.075s for each frame lasting 0.1s, respectively. In addition, for the above three experiments lasting 6.4s, 10s and 10s, respectively, the prposed MNOMP-SPA-KF takes 5.008s, 7.689s, and 7.379s, respectively, which are less than the durations of the data.

% At the moment of $2.5$s, we can see the clear trajectory of people 1 in Fig. \ref{exp3Res_25}.
% The trajectories of car and people 2 is not stable at this moment.
% At the moment of $5$s, $7.5$s and $10$s, we can see the stable trajectories of people 1, people 2 and car, and they move towards their own directions.
% In these experiments, there are many false trajectories, which may caused by the objects in the scenes and the multiple measurements caused by one target.

\section{Conclusion}
In this paper, a systematic solution named MNOMP-SPA-KF algorithm for target sensing and tracking with mmWave Radar is proposed. MNOMP-SPA-KF integrates three modules, namely MNOMP for target estimation and detection, SPA for data association, and KF for target tracking. In particular, MNOMP achieves superresolution, high estimation accuracy, CFAR, and low computational complexity, SPA integrates the radial velocity estimates for data association, and KF takes the measurement correlation into account. Substantial numerical simulations and experiments are conducted to show that the three modules benefits the tracking performance, compared to other benchmark algorithms.

Future work will focus on extending the MNOMP-SPA-KF algorithm for target tracking in 4D mmWave radar systems. Conventional 3D mmWave radars typically provide range, azimuth, and Doppler velocity information, and incorporating velocity information into the Kalman filter (KF) step is expected to significantly enhance tracking accuracy. With the development of 4D mmWave radar \cite{4DmmradarTI}, which can measure an additional elevation angle compared to 3D mmWave radar, there is great potential for improving the performance of target tracking systems. Investigating the adaptation of the MNOMP-SPA-KF algorithm to fully leverage the additional elevation information provided by 4D mmWave radar will be a valuable area of study. This extension could further enhance the accuracy and reliability of target tracking, particularly in complex environments where three-dimensional motion needs to be tracked with greater precision \cite{HuangMT}.

\section{Appendix}
\subsection{The CRB ${\rm CRB}([r, v, \theta]^{\rm T})$}\label{CRBomega}
To derive the CRB, we adopt a two stage procedure.
Firstly, we reparameterize the system model and transform it into a $3$-dimensional LSE problem and the closed-form expression of CRBs of frequencies are derived.
Secondly, by using the vector parameter CRB for transformations, the CRBs of the radial distance, radial velocity and the azimuth are obtained.
Below we describe the details.

Note that in the case of Gaussian observation $\mathbf{y}$, the Fisher information matrix can be calculated via utilizing the following lemma.
\begin{lemma}\label{CRBGussianCase}
\cite{Kayest} Assume that $\mathbf{y} \sim \mathcal{N}({\boldsymbol \eta}(\boldsymbol \xi), {\mathbf C}(\boldsymbol \xi))$, so that both the mean and covariance matrix may depend on ${\boldsymbol \xi}$, the Fisher information matrix is given by
\begin{align}\label{FIM of Gaussian}
[\mathbf{F}(\boldsymbol \xi)]_{ij} = \left[ \frac{\partial {\boldsymbol \eta}({\boldsymbol \xi}) }{\partial \xi_{i}} \right]^{\text{T}} {\mathbf C}^{-1}(\boldsymbol \xi)
\left[ \frac{\partial \boldsymbol{\eta}({\boldsymbol \xi}) }{\partial \xi_{j}} \right] + \frac{1}{2} {\rm tr}\left[{\mathbf C}^{-1}(\boldsymbol \xi) \frac{\partial {\mathbf C}(\boldsymbol \xi)}{\partial \xi_{i}} {\mathbf C}^{-1}(\boldsymbol \xi) \frac{\partial {\mathbf C}(\boldsymbol \xi)}{\partial \xi_{j}} \right],
\end{align}
\end{lemma}

For the single target scenario setting, by vectorizing the observations $\mathcal{Y}$ in (\ref{radarmeas}), the equivalent measurement model
\begin{align}\label{multiple vector model}
\mathbf{y} = \gamma {\mathbf a}({\boldsymbol \omega}) + {\boldsymbol \varepsilon},
\end{align}
can be obtained, where
\begin{align}\label{defineaNomega}
{\mathbf a}({\boldsymbol\omega}) \triangleq &~ \mathbf{a}_{L}({\omega}_{z})\otimes \mathbf{a}_{M}({\omega}_{y}) \otimes\mathbf{a}_{N}({\omega}_{x}),
\end{align}
$\otimes$ denotes the Kronecker product, ${\boldsymbol\varepsilon}={\rm vec}(\mathcal{\varepsilon})$ and $\boldsymbol \varepsilon \sim \mathcal{CN}({\mathbf 0}, \sigma^2 \mathbf{I}_{NML})$. Straightforward calculation shows that $\mathbf{y} \sim \mathcal{CN}(\gamma {\mathbf a}, \sigma^2 \mathbf{I}_{NML})$, which is also equivalent to
\begin{align}\label{Guassian distribution of y}
\begin{bmatrix}
\Re \{ \mathbf{y} \} \\
\Im \{ \mathbf{y} \} \\
\end{bmatrix} \sim \mathcal{N} \left(
\begin{bmatrix}
\Re \{ \gamma {\mathbf a} \} \\
\Im \{ \gamma {\mathbf a} \} \\
\end{bmatrix},  \frac{\sigma^2}{2} \mathbf{I}_{2NML}.
\right)
\end{align}
Thus, in our setting,
$\boldsymbol{\eta}({\boldsymbol \xi})=[\Re \{ \gamma {\mathbf a} \}; \Im \{ \gamma {\mathbf a} \}]$,
the unknown parameter vector $\boldsymbol \xi$ is
\begin{align}\label{xidef}
{\boldsymbol \xi} = [\omega_x, \omega_y, \omega_z, \varphi, g, \sigma^2]^{\rm T},
\end{align}
where $g$ and $\varphi$ is the modulus and phase of the complex amplitude $\gamma$, and the matrix $\mathbf{C}(\boldsymbol{\xi}) = \frac{\sigma^2}{2} \mathbf{I}_{2NML}$.
Substituting ${\boldsymbol \eta}(\boldsymbol{\xi})$, $\boldsymbol{\xi}$ and $\mathbf{C}(\boldsymbol{\xi})$ into  eq.(\ref{FIM of Gaussian}), the Fisher information matrix $\mathbf{F}(\boldsymbol \xi)$ is given by
\begin{align}\label{FIMcaleq}
\left[\mathbf{F}(\boldsymbol \xi)\right]_{ij} =&
\frac{2}{\sigma^2} \left( \left[\frac{\partial \Re\{ \gamma {\mathbf a} \}}{\partial \xi_i}\right]^{\rm T} \left[\frac{\partial \Re\{ \gamma {\mathbf a} \}}{\partial \xi_j}\right] + \left[\frac{\partial \Im\{ \gamma {\mathbf a} \}}{\partial \xi_i}\right]^{\rm T} \left[\frac{\partial \Im\{ \gamma {\mathbf a} \}}{\partial \xi_j}\right]\right) \nonumber \\
&+ \frac{1}{2} {\rm tr}\left[\frac{2}{\sigma^2} \mathbf{I}_{2NML} \frac{\partial \frac{\sigma^2}{2} \mathbf{I}_{2NML} }{\partial \xi_i} \frac{2}{\sigma^2} \mathbf{I}_{2NML} \frac{\partial \frac{\sigma^2}{2} \mathbf{I}_{2NML} }{\partial \xi_j} \right] \nonumber \\
=& \frac{2}{\sigma^2} \Re\left\{ \left[\frac{\partial \gamma {\mathbf a} }{\partial {\boldsymbol \xi}^{\rm T}}\right]^{\rm H} \left[\frac{\partial \gamma {\mathbf a} }{\partial {\boldsymbol \xi}^{\rm T}}\right] \right\} + NML \left(\frac{1}{\sigma^2}\right)^2 \frac{\partial \sigma^2}{\partial \xi_i} \frac{\partial \sigma^2}{\partial \xi_j} \nonumber \\
=& \frac{2}{\sigma^2} \sum_{n,m,l} \left( \Re \left\{ \left[ \frac{\partial [\gamma \mathbf{a}]_{n,m,l}}{\partial {\boldsymbol \xi}^{\rm T}}\right]^{\rm H} \left[ \frac{\partial [\gamma \mathbf{a}]_{n,m,l}}{\partial {\boldsymbol \xi}^{\rm T}}\right]\right\} + \frac{1}{2 \sigma^2} \frac{\partial \sigma^2}{\partial \xi_i} \frac{\partial \sigma^2}{\partial \xi_j}\right)
\end{align}
where $[\gamma \mathbf{a}]_{n,m,l} = g {\rm e}^{{\rm j} \varphi_{n,m,l}}$, $\varphi_{n,m,l} = \varphi + (n - 1) \omega_x + (m - 1) \omega_y + (l - 1) \omega_z$.
And it can be shown that
\begin{align}
\frac{\partial [\gamma \mathbf{a}]_{n,m,l}}{\partial {\boldsymbol \xi}^{\rm T}} =  \left[(n - 1) g {\rm j}, (m - 1) g {\rm j}, (l - 1) g {\rm j}, g {\rm j}, 1, 0\right] {\rm e}^{{\rm j} \varphi_{n,m,l}}.
\end{align}
As a result,
\begin{align}\label{matrixcal}
&\Re \left\{ \left[ \frac{\partial [\gamma \mathbf{a}]_{n,m,l}}{\partial {\boldsymbol \xi}^{\rm T}}\right]^{\rm H} \left[ \frac{\partial [\gamma {\mathbf a}]_{n,m,l}}{\partial {\boldsymbol \xi}^{\rm T}}\right]\right\} = \nonumber \\
&\begin{bmatrix}
(n - 1)^2 g^2 & (m - 1) (n - 1) g^2 & (l - 1)(n - 1) g^2 & (n - 1) g^2 & 0 & 0 \\
(n - 1)(m - 1) g^2 & (m - 1)^2 g^2 & (l - 1)(m - 1) g^2 & (m - 1) g^2 & 0 & 0 \\
(n - 1)(l - 1) g^2 & (m - 1)(l - 1) g^2 & (l - 1)^2 g^2 & (l - 1) g^2 & 0 & 0 \\
(n - 1) g^2 & (m - 1) g^2 & (l - 1) g^2 & g^2 & 0 & 0 \\
0 & 0 & 0 & 0 & 1 & 0 \\
0 & 0 & 0 & 0 & 0 & 0
\end{bmatrix}.
\end{align}
Substituting (\ref{matrixcal}) into (\ref{FIMcaleq})  yields the Fisher information matrix
\begin{align}
\mathbf{F}(\boldsymbol \xi) = \frac{2 NML g^2}{\sigma^2}
\begin{bmatrix}
\mathbf{\Phi} & \mathbf{O} \\
\mathbf{O}^{\rm T} & {\rm diag}\{1, \frac{1}{2 \sigma^2}\} \\
\end{bmatrix}
\end{align}
where $\mathbf{O}$ is the $4 \times 2$ zero-matrix and the matrix $\mathbf{\Phi}$ is
\begin{align}
\mathbf{\Phi} =
\begin{bmatrix}
\frac{1}{6} (N - 1)(2N - 1) & \frac{1}{4} (M - 1)(N - 1) & \frac{1}{4} (L - 1)(N - 1) & \frac{1}{2}(N - 1) \\
\frac{1}{4} (N - 1)(M - 1) & \frac{1}{6} (M - 1)(2M - 1) &\frac{1}{4} (L - 1)(M - 1) & \frac{1}{2}(M - 1) \\
\frac{1}{4} (N - 1)(L - 1) & \frac{1}{4} (M - 1)(L - 1) & \frac{1}{6} (L - 1)(2L - 1) & \frac{1}{2}(L - 1)\\
\frac{1}{2}(N - 1) & \frac{1}{2}(M - 1) & \frac{1}{2}(L - 1) & 1
\end{bmatrix}.
\end{align}
Thus, the CRB of variable ${\boldsymbol \xi}$ is the inverse of the Fisher information matrix, i.e.,
\begin{align}
{\rm CRB}({\boldsymbol \xi}) = \mathbf{F}^{-1}(\boldsymbol \xi) = \frac{\sigma^2}{2 NML g^2}
\begin{bmatrix}
\mathbf{\Phi}^{-1} & \mathbf{O}^{\rm T} \\
\mathbf{O} & {\rm diag}\{1, 2\sigma^2\}
\end{bmatrix}
\end{align}
where the inverse of $\mathbf{\Phi}$ can be calculated as
\begin{align}
\mathbf{\Phi}^{-1} =
\begin{bmatrix}
\frac{12}{(N^2 - 1)} & 0 & 0 & -\frac{6}{(N + 1)} \\
0 & \frac{12}{(M^2 - 1)} & 0 & -\frac{6}{(M + 1)} \\
0 & 0 & \frac{12}{(L^2 - 1)} & -\frac{6}{(L + 1)} \\
-\frac{6}{(N + 1)} & -\frac{6}{(M + 1)} & -\frac{6}{(M + 1)} & 1 + \frac{3(N - 1)}{N + 1} + \frac{3(M - 1)}{M + 1} + \frac{3(L - 1)}{L + 1} \\
\end{bmatrix}.
\end{align}
%So it yields the Cram\'{e}r-Rao bound of angular frequencies as
%\begin{subequations}\label{}
%\begin{align}
%{\rm CRB}(\omega_x) &= \frac{12 \sigma^2}{(N^2 - 1) NML g^2}, \\
%{\rm CRB}(\omega_y) &= \frac{12 \sigma^2}{(M^2 - 1) NML g^2}, \\
%{\rm CRB}(\omega_z) &= \frac{12 \sigma^2}{(L^2 - 1) NML g^2}.
%\end{align}
%\end{subequations}
Therefore, the CRB ${\rm CRB}([\omega_x, \omega_y, \omega_z]^{\rm T})$ can be obtained as
\begin{align}\label{CRBomegaxz}
{\rm CRB}([\omega_x, \omega_y, \omega_z]^{\rm T}) = \frac{6 \sigma^2}{NML g^2}
\begin{bmatrix}
\frac{1}{(N^2 - 1)} & 0 & 0 \\
0 & \frac{1}{(M^2 - 1)} & 0 \\
0 & 0 & \frac{1}{(L^2 - 1)}
\end{bmatrix}
\end{align}

To obtain the CRB ${\rm CRB}([r, v,\theta]^{\rm T})$, the following lemma is utilized.
\begin{lemma} \label{lamma CRB of vector function}
\cite{Kayest} Assume that it is desired to estimate $\mathbf{t}  = \mathbf{h}(\boldsymbol \nu)$ for an $r$-dimensional function $\mathbf{h}$ and a $p$-element vector $\mathbf{t} $. Then we can get that
\begin{align}\label{CRB of vector function}
\mathbf{CRB}({\mathbf{t}}) = \frac{\partial \mathbf{h}(\boldsymbol \nu) }{\partial {\boldsymbol \nu}^{\rm T}} \mathbf{F}^{-1}(\boldsymbol \nu) \left[\frac{\partial \mathbf{h}(\boldsymbol \nu) }{ \partial {\boldsymbol \nu}^{\rm T}} \right]^{\rm{T}},
\end{align}
where $\mathbf{F}(\boldsymbol \nu)$ is the Fisher infromation matrix of vector $\boldsymbol \nu$ and $\frac{\partial \mathbf{h}(\boldsymbol \nu) }{ \partial  \boldsymbol \nu}$ is the $r \times p$ Jacobian matrix.
\end{lemma}
In our setting, $\mathbf{t} = \left[r, v, \theta \right]^{\rm T}$, $\boldsymbol{\nu} = \left[\omega_{x}, \omega_{y}, \omega_{z}\right]^{\rm T}$, the relationship between $\mathbf{t}$ and $\boldsymbol{\nu}$ can be described through a vector function $\mathbf{h}(\boldsymbol{\nu})$, i.e.,
\begin{align}
    \left[r, v, \theta \right]^{\rm T} = \mathbf{h}(\left[\omega_{x}, \omega_{y}, \omega_{z}\right]^{\rm T}) = \left[ \frac{r_{\rm max}}{2 \pi }\omega_x, \frac{ v_{\rm max} }{ \pi } \omega_y, \arcsin\left(\frac{\omega_z}{\pi}\right) \right]^{\rm T}
\end{align}
thus, $\frac{\partial \mathbf{h}(\boldsymbol \nu) }{\partial {\boldsymbol \nu}^{\rm T}}$ is
\begin{align}\label{partialhoveromega}
\frac{\partial \mathbf{h}(\boldsymbol \nu) }{ \partial {\boldsymbol \nu}^{\rm T}} =
\begin{bmatrix}
\frac{r_{\rm max}}{2 \pi} & 0 & 0 \\
0 & \frac{ v_{\rm max} }{ \pi } & 0 \\
0 & 0 & \frac{1}{\pi\cos\theta}
\end{bmatrix}
\end{align}
And the inverse of Fisher information matrix is ${\rm CRB}([\omega_x, \omega_y,\omega_z]^{\rm T})$ shown in eq.(\ref{CRBomegaxz}).
Substituting eq.(\ref{partialhoveromega}) and eq.(\ref{CRBomegaxz}) in eq.(\ref{CRB of vector function}), the CRB ${\rm CRB}([r, v,\theta]^{\rm T})$ 
\begin{align}\label{CRBrvtheta}
	{\rm CRB}([r, v,\theta]^{\rm T}) = \frac{6 \sigma^2}{NML g^2}
    \begin{bmatrix}
		\left(\frac{r_{\rm max}}{ 2 \pi }\right)^2 \frac{1}{N^2 - 1} & 0 & 0 \\
		0 & \left(\frac{ v_{\rm max} }{ \pi }\right)^2 \frac{1}{M^2 - 1} & 0 \\
		0 & 0 & \left(\frac{1}{\pi\cos\theta}\right)^2 \frac{1}{L^2 - 1}
	\end{bmatrix}
\end{align}
is obtained.

% \begin{bmatrix}
% 	\frac{\omega_x}{2 \pi }r_{\rm max} \\
% 	\arcsin\left(\frac{\omega_z}{\pi}\right)
% \end{bmatrix},

%With the help of Lamma \ref{lamma CRB of vector function} and eq.(\ref{FreqToState}), we can obtain the CRB of $r$, $v$ and $\theta$
%\begin{subequations}
%\begin{align}
%{\rm CRB}(r) &= \left(\frac{r_{\rm max}}{2 \pi}\right)^2 \frac{12 \sigma^2}{(N^2 - 1) NML g^2} \approx \left(\frac{r_{\rm max}}{2 \pi}\right)^2 \frac{12 \sigma^2}{N^3 ML g^2}, \\
%{\rm CRB}(v) &= \left(\frac{v_{\rm max}}{2 \pi}\right)^2 \frac{12 \sigma^2}{(M^2 - 1) NML g^2} \approx \left(\frac{v_{\rm max}}{2 \pi}\right)^2 \frac{12 \sigma^2}{NM^3L g^2}, \\
%{\rm CRB}(\theta) &= \left(\frac{1}{\pi \cos\theta}\right)^2 \frac{12 \sigma^2}{(L^2 - 1) NML g^2} \approx \left(\frac{1}{\pi \cos\theta}\right)^2 \frac{12 \sigma^2}{NML^3 g^2}.
%\end{align}
%\end{subequations}

\subsection{The CRB ${\rm CRB}([p_x,p_y]^{\rm T})$ and ${\rm CRB}([p_x,p_y,v]^{\rm T})$}\label{CRBpxy}
We calculate the CRB of the measurement $\mathbf{z} = [p_x,p_y]^{\rm T}$ through $\mathbf{t} = [r,\theta]^{\rm T}$ with the help of lemma \ref{lamma CRB of vector function} \cite{Kayest}, which can be shown as
\begin{align}\label{lammaCRBofz}
\mathbf{CRB}(\mathbf{z}) = \frac{\partial \mathbf{g}(\mathbf{t}) }{ \partial  {\mathbf{t}}^{\rm T}} \mathbf{CRB}(\mathbf{t}) \left[\frac{\partial \mathbf{g}(\mathbf{t}) }{\partial {\mathbf{t}}^{\rm T}}\right]^{\rm T},
\end{align}
where the vector function $\mathbf{g}(\mathbf{t})$ is
\begin{align}
    \mathbf{g}(\mathbf{t}) = \left[r \sin \theta, r \cos \theta\right]^{\rm T},
\end{align}
and the Jacobian matrix $\frac{\partial \mathbf{g}(\mathbf{t}) }{ \partial  {\mathbf{t}}^{\rm T}}$ can be shown as
\begin{align}\label{partialg}
\frac{\partial \mathbf{g}(\mathbf{t})}{\partial {\mathbf{t}}^{\rm T}} =
\begin{bmatrix}
\frac{\partial p_{x}}{\partial r} & \frac{\partial p_{x}}{\partial \theta} \\
\frac{\partial p_{y}}{\partial r} &  \frac{\partial p_{y}}{\partial \theta} \\
\end{bmatrix} =
\begin{bmatrix}
\sin \theta & r \cos \theta \\
\cos \theta & -r \sin \theta
\end{bmatrix},
\end{align}
and the CRB of $\mathbf{t} = \left[r, \theta\right]^{\rm T}$ can be extracted from eq.(\ref{CRBrvtheta}).
Substituting ${\rm CRB}\left(\left[r, \theta\right]^{\rm T}\right)$ and eq.(\ref{partialg}) in eq.(\ref{lammaCRBofz}), ${\rm CRB}([p_x,p_y]^{\rm T})$ (\ref{CRBpxpy}) is obtained. In addition, according to eq. (\ref{lammaCRBofz}), ${\rm CRB}([p_x,p_y,v]^{\rm T})$ is 
\begin{align}\label{CRBpxpyv}
{\rm CRB}([p_x,p_y,v]^{\rm T})=\begin{bmatrix}
    {\rm CRB}([p_x,p_y]^{\rm T})& {\mathbf 0}\\
     {\mathbf 0}&\left(\frac{ v_{\rm max} }{ \pi }\right)^2 \frac{1}{M^2 - 1}
\end{bmatrix}.
\end{align}

\subsection{Valid Region} \label{DofVR}
In order to reduce computation complexity, a general method is to set the valid region of measurements by using the prior information of targets. If we only take the position measurements of the targets into account, the valid region 
\begin{align}\label{2DGate}
	\left( {\mathbf{z}_{h}(t)} - \mathbf{H} \mathbf{x}_k(t | t - 1) \right)^{\rm T} {\mathbf S}_k^{-1} \left( {\mathbf{z}_{h}(t)} - \mathbf{H} \mathbf{x}_k(t | t - 1) \right) \leq d_{G,2D}
\end{align}
can be obtained such that the probability of the above event is the gate probability $P_G$, where $d_{G,2D} = F_{\chi^2(2)}^{-1}(P_G)$, and the covariance ${\mathbf S}_k$ corresponding to the predicted covariance matrix of the target is
\begin{align}\label{covarofSk}
	{\mathbf S}_k = \mathbf{H} \boldsymbol{\Sigma}_k(t | t - 1) \mathbf{H}^{\rm T} + \mathbf{R}(t).
\end{align}
As the radial velocity can also be estimated by the 2D-MNOMP, the radial velocity estimates is also taken into consideration in computing the valid region of measurements. Similarly, the valid region of measurement $\left[  p_{x,h}, p_{y,h}, v_h \right]^{\rm T}$ can be written as
\begin{align}
	\left( \left[ p_{x,h}, p_{y,h}, v_h \right]^{\rm T} - \mathbf{H}_{k}^{\prime} {\mathbf x}_k(t | t - 1) \right)^{\rm T} \left( \mathbf{S}_{k}^{\prime} \right)^{-1} \left( \left[ p_{x,h}, p_{y,h}, v_h \right]^{\rm T} - \mathbf{H}_{k}^{\prime} {\mathbf x}_k(t | t - 1) \right) \le d_{G,3D},
\end{align}
where $v_h$ is the radial velocity, $d_{G,3D}$ is the threshold  which can be calculated as $d_{G,3D} = F_{\chi^2(3)}^{-1}(P_G)$, $\mathbf{H}_{k}^{\prime}$ is 
\begin{align}
	\mathbf{H}_{k}^{\prime} &=
	\begin{bmatrix}
	1 & 0 & 0 & 0 \\
	0 & 0 & 1 & 0 \\
	0 & \sin \theta_{k} & 0 & \cos \theta_{k}
	\end{bmatrix}.
\end{align}
In practice, $\theta_{k}$ is unknown but can be estimated by eq.(\ref{pxpytheta}) using the previous measurements when the tracking is stable.

\bibliographystyle{IEEEbib}
\bibliography{strings,refs}

\end{document}